\documentclass[journal, twocolumn]{IEEEtran}

\ifCLASSINFOpdf

\else

\fi

\usepackage{amsmath,amssymb,graphicx,pgfplots,amsthm,algorithm,algpseudocode, mathtools, multirow, pgfplotstable, makecell}

\usepackage{etoolbox}

\ifCLASSOPTIONcompsoc
 \usepackage[caption=false,font=normalsize,labelfont=sf,textfont=sf]{subfig}
\else
 \usepackage[caption=false,font=footnotesize]{subfig}
\fi
\usepackage[hidelinks, draft=false]{hyperref}

\usepackage{tikz, tikz-3dplot}
\usetikzlibrary{arrows.meta, shapes.geometric}

\pgfplotsset{compat=1.14}
\def\IEEEkeywordsname{Keywords}

\DeclareMathOperator{\diag}{diag}

\DeclareMathOperator{\spn}{span}
\DeclareMathOperator*{\argminA}{arg\,min}
\DeclareMathOperator*{\argmaxA}{arg\,max}
\DeclareMathOperator{\pdet}{pdet}

\pgfplotsset{select coords between index/.style={
    x filter/.code={
        \ifnum\coordindex<#1\def\pgfmathresult{}\fi
    }
}, table/search path={./}}

\newcommand{\transp}{{\sf T}}
\newcommand{\abs}[1]{\lvert#1\rvert}
\newcommand{\norm}[1]{\left\Vert#1\right\Vert}
\newcommand{\roundB}[1]{\left(#1\right)}

\newcommand\inner[2]{\langle #1, #2 \rangle}
\newtheorem{prop}{Proposition}
\newtheorem{theorem}{Theorem}
\newtheorem{lemma}{Lemma}
\newtheorem{coro}{Corollary}
\newcommand\numberthis{\addtocounter{equation}{1}\tag{\theequation}}
\DeclarePairedDelimiter{\ceil}{\lceil}{\rceil}

\newcommand{\comment}[1]{}

\newcounter{function}


\newcounter{function*}
\makeatletter
\newenvironment{function*}[1][htb]{
  \let\c@algorithm\c@function
  \renewcommand{\ALG@name}{Function}
  \begin{algorithm*}[#1]
  }{\end{algorithm*}
}
\makeatother

\newcommand{\AlgSP}{SP}
\newcommand{\AlgRand}{WRS}
\newcommand{\AlgDC}{DC}
\newcommand{\AlgAVM}{AVM}
\newcommand{\AlgEDfree}{LSSS}
\newcommand{\AlgGersh}{BS-GDA}
\newcommand{\AlgWils}{WDPP}
\newcommand{\ProjCS}{POCS}

\newcommand{\recF}{\mathcal{F}}
\newcommand{\recR}{\mathcal{R}}
\newcommand{\bwidth}{f}
\newcommand{\lpDsp}{\mathcal{D}}
\newcommand{\lpDma}{\mathbf{D}}
\newcommand{\lpDspm}{\mathcal{D}_m}
\newcommand{\lpDmam}{\mathbf{D}_m}
\newcommand{\lpTildeDspm}{\tilde{\mathcal{D}}_m}
\newcommand{\lpTildeDsp}{\tilde{\mathcal{D}}}
\newcommand{\lpHatDmam}{\hat{\mathbf{D}}_m}
\newcommand{\bwf}{\mathcal{R}}
\newcommand{\bwfm}{\mathcal{R}_m}
\newcommand{\Smiter}{\mathcal{S}_m}
\newcommand{\SmCiter}{\mathcal{S}_m^c}
\newcommand{\softMat}{MATLAB}

\definecolor{blue1}{RGB}{57, 106, 177}
\definecolor{red1}{RGB}{204, 37, 41}
\definecolor{orange1}{RGB}{218,124,48}
\definecolor{green1}{RGB}{62,150,81}
\definecolor{violet1}{RGB}{107,76,154}
\definecolor{lblue1}{RGB}{51,204,255}

\pgfdeclareplotmark{mystar}{
    \node[star,star point ratio=2.25,minimum size=6pt,
          inner sep=0pt,draw=red1,solid,fill=red] {};
}

\pgfplotscreateplotcyclelist{my col}{
blue1, mark=square*\\%
violet1, mark=otimes*\\%
orange1, mark=triangle*\\%
green1, mark=diamond*\\%
red1, mark=mystar\\%
}

\pgfplotscreateplotcyclelist{mydccol}{
blue1, mark=square*\\%
violet1, mark=otimes*\\%
orange1, mark=triangle*\\%
lblue1, mark=diamond*\\%
red1, mark=mystar\\%
}

\pgfplotscreateplotcyclelist{my3col}{
blue1, mark=square*\\%
green1, mark=diamond*\\%
red1, mark=mystar\\%
}

\pgfplotscreateplotcyclelist{myavmcol}{
red1, mark=mystar\\%
}

\begin{document}

\title{Practical graph signal sampling with log-linear size scaling}

\author{Ajinkya~Jayawant\IEEEauthorrefmark{0*} \texttt{jayawant@usc.edu},
        Antonio~Ortega \texttt{ortega@sipi.usc.edu}\\ Ming Hsieh Department of Electrical and Computer Engineering, University of Southern Calfornia\\ 3740 McClintock Ave, Los Angeles, CA 90089, US
        
\thanks{$^*$ Corresponding author
}}

\maketitle

\begin{abstract}
Graph signal sampling is the problem of selecting a subset of representative graph vertices whose values can be used to interpolate missing values on the remaining graph vertices. Optimizing the choice of sampling set using concepts from experiment design can help minimize the effect of noise in the input signal. While many existing sampling set selection methods are  computationally intensive because they require an eigendecomposition, existing eigendecompostion-free methods are still much slower than random sampling algorithms for large graphs. In this paper, through optimizing sampling sets towards the D-optimal objective from experiment design, we propose a sampling algorithm that has complexity comparable to random sampling algorithms, while reaching accuracy similar to existing eigendecomposition-free methods for a broad range of graph types.
\end{abstract}


\IEEEpeerreviewmaketitle

\section{Introduction}
Graphs are a convenient way to represent and analyze data having irregular relationships between data points \cite{shuman2013emerging} and can be useful in a variety of different scenarios, such as characterizing the Web \cite{kumar2000web}, semi-supervised learning \cite{zhu2003semi}, community detection \cite{fortunato2010community}, or  traffic analysis \cite{crovella2003graph}. We call \textit{graph signal} the data associated with the nodes of a graph.
Similar to traditional signals, a smooth graph signal can be sampled by making observations on a few nodes of the graph, so that the signal at the remaining (non-observed) nodes can be estimated \cite{chen2015discrete,anis2015efficient,tanaka2020sampling}. For this we need to choose a set of vertices, $\mathcal{S}$, called the  sampling set, on which we observe the signal values in order to predict signal values on the other vertices (the complement of $\mathcal{S}$, $\mathcal{S}^c$). In the presence of noise, some sampling sets lead to better signal reconstructions than others: the goal of {\it sampling set selection} is to find the best such sampling set. For traditional discrete signals such as images and audio, downsampling by an integer factor often works well because of the implicit ordering and regular spacing in the signals. Such a structure with ordered and evenly spaced out locations of the discretized signal is unavailable for most graph signals. As a result, the best sampling set is also unknown.
The concept of reconstructing sampled graph signals with some accuracy usually relies on the assumption that the underlying signal is smooth. Intuitively this means that signal values for neighboring vertices aren't drastically different. This is a reasonable assumption in a variety of scenarios such as sensor networks modelling temperature distribution, graph signal representing labels in semi-supervised learning, or preferences in social networks. This makes it  possible for us to reconstruct them by knowing a few signal values \cite{pesenson2008sampling}.

A common model for smooth graph signals assumes that most of their energy is localized in the subspace spanned by a subset of eigenvectors of the graph Laplacian or other graph operator \cite{shuman2013emerging}. Thus, the problem of selecting the best sampling set naturally translates to the problem of selecting a submatrix of the matrix of eigenvectors of the graph Laplacian \cite{anis2015efficient}. Specifically, the problem reduces to a row/column subset selection similar to linear measurement sensor selection problem \cite{joshi2008sensor}.
In the graph signal sampling context, several papers leverage this knowledge to propose novel algorithms --- \cite{shomorony2014sampling, chen2015discrete, tsitsvero2016signals, tremblay2017graph, chamon2017greedy, wang2018optimal}. We refer the reader to \cite{tanaka2020sampling} for a recent comprehensive review of the literature on this topic.

However, to solve the graph sampling set selection problem, row/column selection needs to be applied on the matrix of eigenvectors of the graph Laplacian (or those of some other suitable graph operator). The corresponding eigendecomposition is an $O(n^3)$ operation for an $n\times n$ matrix\footnote{In practice if the signal is bandlimited to the lowest $\bwidth{}$ frequencies, only $\bwidth{}$ eigenvectors need to be computed, but even this can be a complex problem (e.g., a signal bandlimited to the top 10\% frequencies of a graph with millions of nodes). For simplicity, we describe these as full decomposition methods, even though in practice only a subset of eigenvectors is needed.}. This makes it impractical for large graphs in machine learning, social networks, and other applications, for which the cost of  eigendecomposition would be prohibitive. Thus, methods that solve this subset selection problem without explicitly requiring eigendecomposition are valuable.

We can classify sampling set selection methods into two main types of approaches,  based on whether they require eigendecomposition or not. Some methods compute the full eigendecomposition \cite{shomorony2014sampling,chen2015discrete,tsitsvero2016signals}, or instead require a sequential eigendecomposition, where one eigenvector is computed at each step \cite{anis2015efficient}. Alternatively, \textit{eigendecomposition-free} methods do not make use of an  eigendecomposition of the Laplacian matrix  \cite{puy2016random, wang2018optimal, tremblay2017graph, sakiyama2019eigendecomposition} and are usually faster.
Weighted Random Sampling (\AlgRand{}) \cite{puy2016random} is the fastest method but provides only guarantees on average performance, which means that it may exhibit poor reconstruction accuracy for specific instances. It also needs more samples to match the reconstruction accuracy of other eigendecompostion-free methods. Among eigendecomposition-free methods discussed in \cite{tanaka2020sampling}, Neumann series based sampling \cite{wang2018optimal} has a higher computational complexity, Binary Search with Gershgorin Disc Alignment (\AlgGersh) \cite{bai2020fast} has low computational complexity for smaller graphs, but cannot compete with \AlgRand{} for large graphs, and Localization operator based Sampling Set Selection (\AlgEDfree{}) \cite{sakiyama2019eigendecomposition} achieves good performance but requires some parameter tuning
to achieve optimal performance. Our proposed method can overcome these limitations: similar to \cite{wang2018optimal,bai2020fast,sakiyama2019eigendecomposition} it is eigendecomposition-free, but it has complexity closer to \AlgRand{}, while requiring fewer parameters to tune than \AlgRand{}.

Other recently proposed sampling algorithms are eigen\-decomposition-free but involve a different setup than what we consider in this paper. For example, the error diffusion sampling algorithm (Algorithm 5  from \cite{parada2019blue}) achieves complexity comparable to \AlgRand{}. However, the sampling set and the number of samples chosen depend on the vertex numbering in the graph, which has to be done independently of the algorithm in question. In  \cite{parada2019blue} no specific vertex numbering suitable for Algorithm 5 was recommended. A random vertex numbering algorithm would be fast but may lead to suboptimal sampling set choices (similar to what may happen with random sampling). Thus,
more research may be needed to identify efficient numbering algorithms.
Note that other blue noise sampling algorithms \cite{parada2019bluejournal} do not require  vertex numbering, they involve distance computations on the graph similar to \AlgDC{} in \cite{jayawant2018distance}.
In contrast, our proposed algorithm, \AlgAVM{}, is independent of the vertex numbering of the graph and does not require distance computations.
As another example, the algorithms proposed in \cite{basirian2017random} and \cite{abramenko2019graph} are designed for sampling clustered piecewise constant graph signals. However, in this paper, we focus on a bandlimited smoothness model for graph signals, with graph topologies not limited to clustered graphs.

To motivate our methods consider first \AlgRand{},
where vertices are sampled with a probability proportional to their
\textit{squared local coherence} \cite{puy2016random}.
However, selecting vertices having the highest coherence may not result in the best sampling set, because some vertices may be ``redundant''
(e.g., if they are close to each other on the graph). Other sampling algorithms   \cite{sakiyama2019eigendecomposition} improve performance by selecting vertices based on importance but avoid the redundancy by minimizing a notion of overlapped area between functions centered on the sampled vertices.

In our preliminary work \cite{jayawant2018distance}, we proposed the Distance-Coherence (\AlgDC) algorithm, which mitigates the effect of redundancy between vertices by adding new vertices to the sampling set only if they are at a sufficient distance on the graph from the previously selected nodes. While this can eliminate redundancy, it has a negative impact on computation cost, since distance computation is expensive. As an alternative, in this paper we propose a novel Approximate Volume Maximization (\AlgAVM) algorithm that replaces the distance computation with a filtering operation. Loosely speaking, our proposed scheme in \AlgAVM{} precomputes squared coherences, as  \cite{puy2016random}, with an additional criterion to maintain separation between selected vertices using a filtering operation. The resulting complexity (see Section \ref{sec:theoryComplexity}) has a log-linear dependence on the number of edges in a connected graph. The log-linear dependence is desired because it is similar to that of \AlgRand{} which is the fastest algorithm in literature that uses spectral information, second only to unweighted random sampling from \cite{puy2016random}.
\AlgAVM{} can also be viewed as an efficient approximation to the D-optimality criterion \cite{atkinson2007optimum}.
In this paper we review the main concepts in \AlgDC{} and introduce \AlgAVM{}, showing that these methods can improve upon existing algorithms in various ways.
Our main contributions are:
\begin{enumerate}
    \item We describe our distance-based sampling \AlgDC{} algorithm  (Section \ref{sec:propAlg}) to illustrate how to balance the frequency and vertex domain information of graphs for sampling. \AlgDC{} provided us with key ideas to develop the \AlgAVM{} algorithm and can potentially serve as the basis for hybrid algorithms.
    \item We introduce a new algorithm, \AlgAVM{}  (Algorithm \ref{alg:avm}), which can be used for any graph size or topology while requiring few parameters to tune. Moreover, the accuracy of the reconstruction is a monotonic function of those parameters. This eliminates the need to search for the right parameter, as we only evolve a parameter unidirectionally for a better reconstruction. \item Using the framework of volume based sampling (Section \ref{sec:propAlg}),
    we interpret a series of algorithms --- exact greedy \cite{tsitsvero2016signals}, \AlgRand{}, Spectral Proxies (\AlgSP{}) \cite{anis2015efficient}, \AlgEDfree{}, \AlgDC{}, and our proposed \AlgAVM{} as variations of the volume maximization problem formulation (Section  \ref{sec:cutoff}), and explain critical differences between existing methods and \AlgAVM{}.
    \item \AlgAVM{} provides competitive reconstruction performance on a variety of graphs and sampling scenarios, improving reconstruction signal-to-noise ratio (SNR) over \AlgRand{} by at least 0.6dB and frequently significantly higher (e.g., 2dB) --- Section \ref{sec:expVol}. The practicality of \AlgAVM{} is apparent for larger graph sizes (e.g., of the order of a hundred thousand nodes)---: with the limits placed by the system used in our experiments(see Section \ref{sec:effectGraphSizes}), other state-of-the-art algorithms such as \AlgSP{}, \AlgEDfree{} and \AlgGersh{} often fail at these graph sizes, while a complete execution is always possible for \AlgAVM{}. At graph sizes small enough  for the other algorithms to be applied, \AlgAVM{} is at least 2.5 times and often orders of magnitude faster compared to state-of-the-art algorithms such as \AlgSP{}, \AlgEDfree{} and \AlgGersh{}, while sacrificing less than 0.01dB of reconstruction SNR --- Section \ref{sec:perfAlg}. We explain these advantages in terms of complexity towards the end of Section \ref{sec:propAlg} by showing that compared to \AlgRand{}, the additional computations needed by \AlgAVM{} scale linearly as a function of the number of edges in a connected graph.
\end{enumerate}

As a summary, our proposed \AlgAVM{} sampling algorithm has complexity comparable to the \AlgRand{} sampling algorithm along with a significantly better reconstruction accuracy. It achieves this without requiring any prior knowledge of the signal bandwidth, and can be used  for different graphs while requiring a few easy-to-tune parameters.

\section{Problem setup}

\begin{table*}[t]
    \centering
    \caption{Linear algebra notation in this paper}
    \begin{tabular}{| c | c |}\hline
        Notation & Description \\ \hline
        $\mathcal{X}_i$ & $\mathcal{X}$ after iteration $i$ \\ \hline
        $\abs{\mathcal{X}}$ & Cardinality of set $\mathcal{X}$ \\ \hline
        $\mathbf{A}_{\mathcal{XY}}$ or $\mathbf{A}_{\mathcal{X,Y}}$ & Submatrix of $\mathbf{A} $ indexed by sets $\mathcal{X}$ and $\mathcal{Y}$ \\ \hline
        $\mathbf{A}_{ij}$ & $(i,j)^\text{th}$ element of $\mathbf{A}$ \\ \hline
        $\mathbf{A}_{\mathcal{X}}$ & $\mathbf{A}_{\mathcal{:,X}}$, selection of the columns of $\mathbf{A}$\\ \hline
        $\mathbf{A}_i$ & $\mathbf{A}$ after  iteration $i$\\ \hline
        $x_i$ or $\mathbf{x}(i)$ & $i^{th}$ element of the vector $\mathbf{x}$\\ \hline
        $\mathbf{x}_{\mathcal{X}}$ or $\mathbf{x}(\mathcal{X})$ & Subset of the vector $\mathbf{x}$ corresponding to indices $\mathcal{X}$\\ \hline
        $\mathbf{x}_v$ & Vector  corresponding to a vertex $v$ among a sequence of vectors indexed over the set of vertices $\mathcal{V}$\\ \hline
        $\norm{.}$ & Two/Euclidean norm of matrix or vector
        \\ \hline
        $\abs{x}, \abs{\mathbf{x}}$ & Entry wise absolute value of scalar $x$ or vector $\mathbf{x}$\\ \hline
    \end{tabular}
    \label{tab:papNote}
\end{table*}

\subsection{Notation}
In this paper, we represent sets using calligraphic uppercase, e.g., $\mathcal{X}$, vectors using bold lowercase,  $\mathbf{x}$, matrices using bold uppercase, $\mathbf{A}$, and scalars using plain uppercase or lowercase as $x$ or $X$. Table \ref{tab:papNote} lists additional notations.

A graph is defined as the pair $(\mathcal{V},\mathcal{E})$, where $\mathcal{V}$ is the set of nodes or vertices and $\mathcal{E}$ is the set of edges \cite{bollobas2013modern}. The set of edges $\mathcal{E}$ is a subset of the set of unordered pairs of elements of $\mathcal{V}$. A graph signal is a real-valued function defined on the vertices of the graph, $\mathbf{f}:\mathcal{V}\rightarrow \mathbb{R}$.
We index the vertices $v\in\mathcal{V}$ with the set $\{1,\cdots,n\}$ and define $w_{ij}$
as the weight of the edge between vertices $i$ and $j$. The $(i,j)^\text{th}$ entry of the adjacency matrix of the graph $\mathbf{A}$ is $w_{ij}$, with
$w_{ii} =0$, where $n$ is the number of vertices in the graph, which we also call as the graph size.
The degree matrix $\mathbf{D}$ of a graph is a diagonal matrix with diagonal entries $d_{ii}= \sum_j w_{ij}$.
In this paper we consider weighted undirected graphs, without self loops and with non-negative edge weights. Throughout the paper,
$\mathbf{I}$ is $n \times n$ identity matrix.

The combinatorial Laplacian for the graph is given by $\mathbf{L} = \mathbf{D} - \mathbf{A}$, with its corresponding eigendecomposition defined as  $\mathbf{L} = \mathbf{U} \Sigma \mathbf{U}^\transp$ since the Laplacian matrix is symmetric and positive semidefinite.  The eigenvalues of the Laplacian matrix are $\boldsymbol{\Sigma}=\diag(\lambda_1,\cdots,\lambda_n)$, with $\lambda_1 \leq \cdots \leq \lambda_n$ representing the frequencies. The column vectors of $\mathbf{U}$ provide a frequency representation for graph signals, so that the operator $\mathbf{U}^\transp$ is usually called the graph Fourier transform (GFT).
The eigenvectors $\mathbf{u}_i$ of $\mathbf{L}$ associated with larger eigenvalues $\lambda_i$ correspond to higher frequencies, and the ones associated with lower eigenvalues correspond to lower frequencies \cite{shuman2013emerging}.

The sampling set ${\mathcal S}$ is defined as a subset of $\mathcal{V}$ where the values of the graph signal $\mathbf{f}$ are known, leading to a vector of known values $\mathbf{f}_{\mathcal S}$. The problem we consider here is that of finding the set ${\mathcal S}$ such that the error in interpolating  $\mathbf{f}_{\mathcal{S}^c}$ from $\mathbf{f}_{\mathcal S}$ is minimized. Here, different error metrics are possible and the actual error depends on assumptions made about the signal. When comparing algorithms we assume they all operate with the same sampling set size: $s$. For the sake of convenience, without loss of generality, for a given algorithm the vertices are relabeled after sampling, so that their labels correspond to the order in which they were chosen, $\mathcal{S} = \{1, 2, \cdots\}$.

For reconstruction, we will often work with sub-matrices of $\mathbf{U}$ corresponding to different frequencies or vertex localizations.
The cardinality of the set of frequencies, $\abs{\recF{}}$, is the bandwidth of the signal, whereas the set $\recF{}$ is  the bandwidth support. Letting $\recF{}$ be the set $\{1,\cdots,\bwidth{}\}$, where $\bwidth{}=\abs{\recF{}}$, the matrix constructed by selecting the first $\bwidth{}$ columns of $\mathbf{U}$ will be denoted by $\mathbf{U}_{\mathcal{VF}}$ or simply $\mathbf{U}_\recF{}$. The matrix constructed by further selecting rows of $\mathbf{U}_\recF{}$ indexed by $\mathcal{S}$ (corresponding to selected nodes) will be written as $\mathbf{U}_{\mathcal{S}\recF{}}$.

\subsection{Problem formulation}
\label{sec:probForm}
For sampling bandlimited signals $\mathbf{x}$, which can be written as
\begin{equation*}
    \mathbf{x} = \mathbf{U}_\recF{} \tilde{\mathbf{x}}_\recF{},
\end{equation*}
a sampling set that satisfies the following two conditions: i) the number of samples requested is larger than the bandwidth, that is $\abs{\mathcal{S}}\geq \bwidth{}$, and ii) the sampling set $\mathcal{S}$ is a uniqueness set \cite{pesenson2008sampling} corresponding to the bandwidth support $\recF{}$,  will allow us to recover $\mathbf{x}$ exactly.
Given the observed samples, $\mathbf{x}_\mathcal{S}$, the reconstruction is given by the least squares solution:
\begin{equation}
    \mathbf{\hat{x}} = \mathbf{U}_\recF{} (\mathbf{U}_{\mathcal{S}\recF{}}^\transp \mathbf{U}_{\mathcal{S}\recF{}})^{-1} \mathbf{U}_{\mathcal{S}\recF{}}^\transp\mathbf{x}_\mathcal{S}.
    \label{eq:recon}
\end{equation}

In this paper we consider the widely studied scenario of bandlimited signals with added noise, and choose sampling rates that satisfy Condition i) for the underlying noise-free signal\footnote{We do not consider cases where signals are not bandlimited but can be sampled and reconstructed (refer to  \cite{tanaka2020sampling} and references therein). Exploring more general models for signal sampling is left for future work.}. While Condition ii) is difficult to verify without computing the eigendecomposition of the Laplacian, it is likely to be satisfied if Condition i) holds.
Indeed, for most graphs, except those that are either disconnected or have some symmetries (e.g., unweighted path or grid graphs),  any sets such that $\abs{\mathcal{S}}\geq \bwidth{}$ are uniqueness sets,
 Thus, similar to most practical sampling methods \cite{anis2015efficient, puy2016random,sakiyama2019eigendecomposition, bai2020fast},  our sampling algorithms are not designed to return uniqueness sets satisfying Condition ii) thus providing exact recovery, and instead we assume that Condition i) is sufficient to guarantee exact recovery.

In practice signals are never exactly bandlimited and it is common to consider the signal model $\mathbf{f} = \mathbf{x} + \mathbf{n}$, where $\mathbf{x}$ is bandlimited and $\mathbf{n}$ is a noise vector. The reconstruction from the sampled signal $\mathbf{f}_\mathcal{S}=\mathbf{x}_\mathcal{S}+\mathbf{n}_\mathcal{S}$ is then:
\begin{equation*}
    \mathbf{\hat{f}} = \mathbf{U}_\recF{} (\mathbf{U}_{\mathcal{S}\recF{}}^\transp \mathbf{U}_{\mathcal{S}\recF{}})^{-1} \mathbf{U}_{\mathcal{S}\recF{}}^\transp(\mathbf{x}_\mathcal{S}+\mathbf{n}_\mathcal{S}).
\end{equation*}
Since \eqref{eq:recon} allows us to reconstruct $\mathbf{x}$ exactly, the error in the reconstructed signal is:
\begin{equation*}
    \mathbf{\hat{f}} - \mathbf{x} = \mathbf{U}_\recF{} (\mathbf{U}_{\mathcal{S}\recF{}}^\transp \mathbf{U}_{\mathcal{S}\recF{}})^{-1} \mathbf{U}_{\mathcal{S}\recF{}}^\transp\mathbf{n}_\mathcal{S}.
\end{equation*}
The expected value of the corresponding error matrix, $(\mathbf{\hat{f}-\mathbf{x}})\allowbreak(\mathbf{\hat{f}}-\mathbf{x})^\transp$, is
\begin{align*}
&\mathbb{E}[(\mathbf{\hat{f}-\mathbf{x}})(\mathbf{\hat{f}}-\mathbf{x})^\transp] = \notag \\  &\mathbf{U}_\recF{} (\mathbf{U}_{\mathcal{S}\recF{}}^\transp \mathbf{U}_{\mathcal{S}\recF{}})^{-1} \mathbf{U}_{\mathcal{S}\recF{}}^\transp E[\mathbf{n}_\mathcal{S} \mathbf{n}_\mathcal{S}^\transp]\mathbf{U}_{\mathcal{S}\recF{}} (\mathbf{U}_{\mathcal{S}\recF{}}^\transp \mathbf{U}_{\mathcal{S}\recF{}})^{-1} \mathbf{U}_\recF{}^\transp.
\end{align*}
If we assume individual noise entries to be independent with zero mean and equal variance, the expected value, which is the error covariance matrix becomes
\begin{equation}
\mathbb{E}[(\mathbf{\hat{f}-\mathbf{x}})(\mathbf{\hat{f}}-\mathbf{x})^\transp] = c \mathbf{U}_\recF{} (\mathbf{U}_{\mathcal{S}\recF{}}^\transp \mathbf{U}_{\mathcal{S}\recF{}})^{-1} \mathbf{U}_\recF{}^\transp
\label{eq:covariance}
\end{equation}
for a constant $c$. Different metrics of the reconstruction error $\mathbf{\hat{f}} - \mathbf{x}$ can be optimized by maximizing a function $h:M_{n,n}(\mathbb{R})\rightarrow \mathbb{R}$ of the error covariance matrix, where $M_{n,n}(\mathbb{R})$ is an $n\times n$ matrix of real numbers. Since the error covariance matrix is a function of the sampling set $\mathcal{S}$, we wish to find an $\mathcal{S}$ that maximizes a function $h(.)$ of the error covariance matrix as follows:
\begin{equation}
    \mathcal{S} = \argmaxA_{\mathcal{S}\subset \mathcal{V},\abs{\mathcal{S}}=s} h\roundB{\mathbf{U}_{\recF{}}(\mathbf{U}_{\mathcal{S}\recF{}}^\transp\mathbf{U}_{\mathcal{S}\recF{}})^{-1}\mathbf{U}_{\recF{}}^\transp}.
    \label{eq:objMat}
\end{equation}
Note that the set $\mathcal{S}$ achieving optimality under general criteria in the form of \eqref{eq:objMat}  is a function of  $\recF{}$, so that $\mathcal{S}$ is optimized for reconstruction with that particular bandwidth support $\recF{}$.
While typically we do not know the bandwidth of the original signal, in what follows we assume that a specific bandwidth for reconstructing the signal has been given.

A particular choice $h(.)$ of interest to us is $1/\pdet(.)$, where $\pdet(.)$ is the pseudo determinant\cite{minka1998inferring}. Since our error covariance matrix is singular, we used pseudo determinant instead of determinant. Pseudo determinant only differs from determinant in that it is a product of non-zero eigenvalues instead of all eigenvalues of the matrix. With our choice of $h(.)$, \eqref{eq:objMat} is equivalent to the following maximization :
\begin{equation}
    \mathcal{S} = \argmaxA_{\mathcal{S}\subset \mathcal{V},\abs{\mathcal{S}}=s} \det(\mathbf{U}_{\mathcal{S}\recF{}}^\transp \mathbf{U}_{\mathcal{S}\recF{}}).
    \label{eq:mainObj}
\end{equation}
This is also known as the D-optimality criterion. Maximizing the determinant leads to minimizing the confidence interval of the solution $\mathbf{\hat{f}}$ \cite{atkinson2007optimum}
as will be seen in Appendix \ref{sec:ignoreBand}.
As a further advantage, the D-optimal objective leads to a novel unified view of different types of sampling algorithms proposed in the literature --- see Section \ref{sec:tradeoffInn}.
Moreover, the D-optimal objective is necessary for the approximations we need in order to develop algorithms achieving eigendecomposition-free subset selection.

Sampling algorithms are designed to implicitly or explicitly optimize the sampling set for a particular bandwidth support. In this paper, we denote by $\recR{}$ the bandwidth support assumed by a sampling algorithm, which can be equal to the  reconstruction bandwidth support $\recF{}$ for which the objective \eqref{eq:mainObj} can be rewritten as:
\begin{equation}
    \mathcal{S} = \argmaxA_{\mathcal{S}\subset \mathcal{V},\abs{\mathcal{S}}=s} \det(\mathbf{U}_{\mathcal{S}\recR{}}^\transp \mathbf{U}_{\mathcal{S}\recR{}}), \;\; \text{with} \;\;  \qquad \recR{}=\recF{}.
    \label{eq:FreplaceR}
\end{equation}
However, there are advantages to choosing a different $\mathcal{R}$ for optimization than $\recF{}$. For example, if we consider $\recR{}=\{1,\cdots,s\}$ so that $\abs{\mathcal{R}}=\abs{\mathcal{S}}$, we can rewrite the objective function \eqref{eq:FreplaceR} without changing its value, by permuting the order of the matrices:
\begin{equation}
    \mathcal{S} = \argmaxA_{\mathcal{S}\subset \mathcal{V},\abs{\mathcal{S}}=s} \det(\mathbf{U}_{\mathcal{S}\recR{}}\mathbf{U}_{\mathcal{S}\recR{}}^\transp).
    \label{eq:doptObj}
\end{equation}
Essentially, instead of using the reconstruction frequency $f$ as the sampling frequency, we use the number of samples requested, $s$, as a proxy for the sampling frequency. As we will see, this new form of \eqref{eq:doptObj} is easier to interpret and use.

Since choosing  $\abs{\mathcal{R}}=\abs{\mathcal{S}}$ is required, it raises concerns about the optimality of our sampling set for the original objective function. This issue will be discussed in Appendix \ref{sec:ignoreBand}.

\subsection{Solving D-optimal objectives}
\label{sec:doptObj}

D-optimal subsets for matrices are determinant maximizing subsets. The determinant measures the volume, and selecting a maximum volume submatrix is an NP-Hard problem \cite{ccivril2009selecting}. Nearly-optimal methods have been proposed in the literature  \cite{goreinov2010find}, \cite{deshpande2010efficient}, but these are based on selecting a submatrix of rows or columns of a known matrix.
Similarly, in the graph signal processing literature, several contributions
\cite{tsitsvero2016signals,chamon2017greedy} develop algorithms for D-optimal selection assuming that $\mathbf{U}$ is available.
In contrast, the main novelty of our work is to develop greedy algorithms for approximate D-optimality, i.e., solving  \eqref{eq:mainObj} without requiring explicit eigendecomposition to obtain $\mathbf{U}$. This is made possible by specific characteristics of our problem to be studied next.

Among graph signal sampling approaches that solve the D-optimal objective,
the closest to our work is the application of Wilson's algorithm for Determinantal Point Process (\AlgWils{})  of  \cite{tremblay2017graph}, which similarly does not require explicitly computing $\mathbf{U}$ . However, our proposed technique, \AlgAVM{}, achieves this goal in a different way and leads to better performance.
Specifically,  \AlgWils{} avoids eigendecomposition while approximating the bandlimited kernel using Wilson's marginal kernel \cite{tremblay2017graph}
upfront. This is a one-time approximation, which does not have to be updated each time nodes are added to the sampling set. This approach relies on a relation between Wilson's marginal kernel and random walks on the graph, leading to a probability of choosing sampling sets that is proportional to the determinant \cite{tremblay2017graph}. In contrast, \AlgAVM{} solves an approximate optimization at each iteration, i.e., each time a new vertex is added to the existing sampling set.
Thus, \AlgAVM{} optimizes the cost function \eqref{eq:mainObj} at every iteration as opposed to \AlgWils{} which aims to achieve the expected value of the cost function.

The \AlgWils{} and \AlgAVM{} algorithms differ in their performance as well. \AlgAVM{} is a greedy algorithm, and the performance greedy determinant maximization algorithms is known to lie within a factor of the maximum determinant \cite{ccivril2009selecting}. In contrast, \AlgWils{} samples with probabilities proportional to the determinants, so that its average performance depends on the distribution of the determinants.
In fact, for certain graph types in \cite{tremblay2017graph}, we observe that \AlgWils{} has worse average performance than \AlgRand{}.
In comparison,
in our experiments, for a wide variety of graph topologies and sizes, \AlgAVM{} consistently outperforms  \AlgRand{}  \cite{puy2016random} in terms of average reconstruction error.

\section{Efficient sampling set selection algorithms}
\label{sec:propAlg}

In what follows we assume that the conditions for equivalence between the two objective function forms  \eqref{eq:FreplaceR} and \eqref{eq:doptObj} are verified, so that we focus on solving \eqref{eq:doptObj}.

\subsection{Incremental subset selection}

The bandwidth support for the purpose of sampling is assumed to be $\mathcal{R}=\{1,\cdots,s\}$. Let us start by defining the low pass filtered signal for the Kronecker delta function $\boldsymbol{\delta}_v$ localized at vertex $v$:
\begin{equation}
\mathbf{d}_v = \mathbf{U}_{\mathcal{R}}\mathbf{U}_{\mathcal{R}}^\transp \boldsymbol{\delta}_v.
\label{eq:lpdel}
\end{equation}
With this definition, the objective in \eqref{eq:doptObj} can be written as:
\begin{align*}
    \det(\mathbf{U}_{\mathcal{S}\recR{}}\mathbf{U}_{\mathcal{S}\recR{}}^\transp ) &= \det(\mathbf{U}_{\mathcal{S}\recR{}} \mathbf{U}_{\mathcal{R}}^\transp \mathbf{U}_{\mathcal{R}}\mathbf{U}_{\mathcal{S}\recR{}}^\transp)\\
    &= \det(\mathbf{I}_{\mathcal{S}}^\transp\mathbf{U}_{\mathcal{R}} \mathbf{U}_{\mathcal{R}}^\transp \mathbf{U}_{\mathcal{R}}\mathbf{U}_{\mathcal{R}}^\transp\mathbf{I}_{\mathcal{S}})\\
    &= \det\roundB{\begin{bmatrix}\mathbf{d}_1&\cdots&\mathbf{d}_s\end{bmatrix}^\transp \begin{bmatrix}\mathbf{d}_1&\cdots&\mathbf{d}_s\end{bmatrix}}\\
    &= \text{Vol}^2(\mathbf{d}_1,\cdots,\mathbf{d}_s).\numberthis\label{eq:detEqui}
\end{align*}
Here $\mathbf{I}_{\mathcal{S}}$ represents the submatrix obtained by selecting the columns of $\mathbf{I}$ indexed by set $\mathcal{S}$. Thus, maximizing the determinant $\det(\mathbf{U}_{\mathcal{S}\recR{}}\mathbf{U}_{\mathcal{S}\recR{}}^\transp )$ is equivalent to maximizing $\text{Vol}(\mathbf{d}_1,\allowbreak\cdots,\mathbf{d}_s)$, and as a consequence the set maximizing \eqref{eq:detEqui} also maximizes \eqref{eq:doptObj}.

In an iterative algorithm where the goal is to select $s$ samples, consider a point where $m<s$ samples have been selected and we have to choose the next sample from among the remaining vertices. Throughout the rest of the paper, we denote the sampling set at the end of the $m^\text{th}$ iteration of an algorithm by $\Smiter{}$.
Given the first $m$ chosen samples we define  $\lpDmam{} = [\mathbf{d}_1 \cdots  \mathbf{d}_m ]$ and the space spanned by the vectors in $\lpDmam{}$ as $\lpDspm{} = \spn(\mathbf{d}_1, \cdots, \mathbf{d}_m)$. Throughout the rest of the paper, we denote $\lpDsp{}$ and $\lpDma{}$ at the end of the $m^\text{th}$ iteration of an algorithm by $\lpDspm{}$ and $\lpDmam{}$. Note that both $\lpDspm{}$ and $\lpDmam{}$ are a function of the choice of the sampling bandwidth support $\recR{}$.
Next, the best column $\mathbf{d}_v$ to be added to $\lpDmam{}$ should maximize:
\begin{subequations}\label{eq:detObj}
\begin{align}
    &\det\roundB{\begin{bmatrix} \lpDmam{} & \mathbf{d}_v\end{bmatrix}^\transp \begin{bmatrix} \lpDmam{} &
    \mathbf{d}_v\end{bmatrix}}
    \\&
    =\det\roundB{\begin{bmatrix} \lpDmam^\transp \lpDmam{} & \lpDmam^\transp \mathbf{d}_v\\ \mathbf{d}_v^\transp \lpDmam{} & \mathbf{d}_v^\transp \mathbf{d}_v\end{bmatrix}} \\&= \det(\lpDmam^\transp \lpDmam{})\det(\mathbf{d}_v^\transp \mathbf{d}_v - \mathbf{d}_v^\transp \lpDmam{} (\lpDmam^\transp\lpDmam{})^{-1}\lpDmam^\transp \mathbf{d}_v)\label{eq:schurUpdate}\\
    &= \det(\lpDmam^\transp \lpDmam{})(\norm{\mathbf{d}_v}^2 - \mathbf{d}_v^\transp \lpDmam{} (\lpDmam^\transp\lpDmam{})^{-1}\lpDmam^\transp \mathbf{d}_v)\label{eq:multUpdate}\\
    &= \det(\lpDmam^\transp \lpDmam{})\roundB{\norm{\mathbf{d}_v}^2 - \norm{\mathbf{P}_{\lpDspm{}} \mathbf{d}_v}^2}. \label{eq:projInterpret}
\end{align}
\end{subequations}
The effect on the determinant of adding a column to $\lpDmam{}$  can be represented according to a multiplicative update (Section 11.2 \cite{atkinson2007optimum}) in our D-optimal design.  \eqref{eq:schurUpdate} follows from \cite{horn2012matrix} (Section 0.8.5 in Second Edition), while  \eqref{eq:projInterpret} follows because $\mathbf{P}_{\lpDspm{}} = \lpDmam{} (\lpDmam^\transp\lpDmam{})^{-1}\lpDmam^\transp$ is a projection onto the space $\lpDspm{}$.
Direct greedy determinant maximization requires selecting a vertex that maximizes the update term in \eqref{eq:multUpdate}:
\begin{equation}
    v^* = \argmaxA_{v\in \SmCiter{}} \left( \norm{\mathbf{d}_v}^2 - \mathbf{d}_v^\transp \lpDmam{} (\lpDmam^\transp\lpDmam{})^{-1}\lpDmam^\transp \mathbf{d}_v  \right)
    \label{eq:objMultUpdate}
\end{equation}
over all possible vertices $v\in \SmCiter{}$, which requires the expensive computation of $(\lpDmam^\transp\lpDmam{})^{-1}$.

The first step towards a greedy incremental vertex selection is estimating the two components, $\norm{\mathbf{d}_v}^2$ and $\mathbf{d}_v^\transp \lpDmam{}\allowbreak (\lpDmam^\transp\lpDmam{})^{-1}\lpDmam^\transp \mathbf{d}_v$, of the multiplicative update. The first term $\norm{\mathbf{d}_v}^2$  is the \textit{squared coherence} introduced in  \cite{puy2016random}, which is estimated here using the same techniques as in \cite{puy2016random}, and is defined as
\begin{equation}
    \norm{\mathbf{d}_v}^2 = \norm{\mathbf{U}_\mathcal{R}\mathbf{U}_\mathcal{R}^\transp \boldsymbol{\delta}_v}^2 = \norm{\mathbf{U}_\mathcal{R}^\transp \boldsymbol{\delta}_v}^2.
    \label{eq:cohDef}
\end{equation}
For the second term, the projection interpretation of \eqref{eq:projInterpret} will be useful to develop approximations to reduce complexity.
Additionally, we will make use of the following property of our bandlimited space to develop an approximation.
\begin{lemma}
The space of bandlimited signals $\spn(\mathbf{U}_{\mathcal{R}})$ equipped with the dot product is a reproducing kernel Hilbert space (RKHS).
\end{lemma}
\begin{proof}

Defining the inner product for signals $\mathbf{f},\mathbf{g}\in \spn(\allowbreak\mathbf{U}_{\mathcal{R}})$ as $\inner{\mathbf{f}}{\mathbf{g}}=\sum_i f_i g_i$, $\spn(\mathbf{U}_{\mathcal{R}})$ is a Hilbert space.
A Hilbert space further needs an existing reproducing kernel to be an RKHS. Towards that end, consider a mapping to our bandlimited space  $\phi:\mathbb{R}^n \rightarrow \spn(\mathbf{U}_{\mathcal{R}})$ given as:
\begin{equation*}
    \phi(\mathbf{x}) = \mathbf{U}_{\mathcal{R}}\mathbf{U}_{\mathcal{R}}^\transp \mathbf{x}.
\end{equation*}
A function $K:\mathbb{R}^n \times \mathbb{R}^n \rightarrow \mathbb{R}$ that uses that mapping and the scalar product in our Hilbert space is:
\begin{equation*}
    \forall \mathbf{x}, \mathbf{y} \in \mathbb{R}^n,\qquad K(\mathbf{x},\mathbf{y}) = \inner{\phi(\mathbf{x})}{\phi(\mathbf{y})}.
\end{equation*}
Now using Theorem 4 from \cite{berlinet2011reproducing}, $K$ is a reproducing kernel for our Hilbert space and using Theorem 1 from \cite{berlinet2011reproducing} we conclude that our bandlimited space of signals is an RKHS.
\end{proof}

\begin{coro}
The dot product of a bandlimited signal $\mathbf{f}\in \spn(\mathbf{U}_{\mathcal{R}})$ with a filtered  delta $\mathbf{d}_v$ is $\mathbf{f}(v)$, the entry at node $v$ of signal $\mathbf{f}$:
\label{prop:rkhsDot}
\begin{equation}
    \inner{\mathbf{f}}{\mathbf{d}_v} = \mathbf{f}(v).
    \label{eq:rkDot}
\end{equation}
\end{coro}
\begin{proof}

The dot product $\inner{\mathbf{f}}{\mathbf{d}_v}$ in our RHKS can be seen as the evaluation functional of $\mathbf{f}$ at the point $v$. Using the definition of reproducing kernel $K$, the evaluation functional for a signal $\mathbf{f}$ in the bandlimited space at a point $\mathbf{x}\in \mathbb{R}^n$ is (using Section 2 definition and Theorem 1 Property b from \cite{berlinet2011reproducing}) $\inner{\mathbf{f}}{\phi(\mathbf{x})}$. This definition provides us the required interpretation for $\inner{\mathbf{f}}{\mathbf{d}_v}$:
\begin{equation}
    \inner{\mathbf{f}}{\mathbf{d}_v} = \inner{\mathbf{f}}{ \phi(\boldsymbol{\delta}_v)}.
    \label{eq:fevalInterpret}
\end{equation}

An evaluation functional $\inner{\mathbf{f}}{\phi(\mathbf{x})}$ in our bandlimited space can be simplified as:
\begin{align}
    \inner{\mathbf{f}}{\phi(\mathbf{x})} = \inner{\mathbf{f}}{\mathbf{U}_{\mathcal{R}}\mathbf{U}_{\mathcal{R}}^\transp \mathbf{x}}\nonumber
    &= \inner{\mathbf{U}_{\mathcal{R}}\mathbf{U}_{\mathcal{R}}^\transp \mathbf{f}}{ \mathbf{x}}\nonumber\\
    &= \inner{\mathbf{f}}{ \mathbf{x}}\label{eq:evalDot}.
\end{align}
Thus, from \eqref{eq:fevalInterpret} and \eqref{eq:evalDot}:
\[
\inner{\mathbf{f}}{\mathbf{d}_v} = \inner{\mathbf{f}}{ \phi(\boldsymbol{\delta}_v)} = \inner{\mathbf{f}}{\boldsymbol{\delta}_v} = \mathbf{f}(v).
\]
\end{proof}
As a consequence of \eqref{eq:rkDot}, if $\mathbf{f}=\mathbf{d}_w$ we have:
\begin{equation}
    \inner{\mathbf{d}_w}{\mathbf{d}_v} = \mathbf{d}_v(w) = \mathbf{d}_w(v).
    \label{eq:sideEffect}
\end{equation}

\subsection{Approximation through distances}
\label{sec:distApprox}

We start by proposing a distance based algorithm (\AlgDC) based on the updates we derived in \eqref{eq:detObj}. While in principle  those updates are valid only when  $s=\bwidth{}$, in \AlgDC{}
we apply them even when $s>\bwidth{}$. We assume $\bwidth{}$ is known to us. We take the bandwidth support for the purpose of sampling to be $\mathcal{R}=\{1,\cdots,f\}$, which is the same as the signal reconstruction bandwidth support $\recF{}$.
To maximize the expression in \eqref{eq:projInterpret} we would like to select nodes that have:
\begin{enumerate}
    \item Large squared local squared graph coherence $\norm{\mathbf{d}_v}^2$ with respect to $\bwidth{}$ frequencies (the first term in \eqref{eq:projInterpret}, which is a property of each node and independent of $\Smiter{}$), and
    \item small squared magnitude of projection onto the subspace $\lpDspm{}$ (which does depend on $\Smiter{}$) $\norm{\mathbf{P}_{\lpDspm{}} \hat{\mathbf{d}_v}}^2$ and thus would increase \eqref{eq:projInterpret}.
\end{enumerate}

The squared local graph coherence \eqref{eq:cohDef} of a vertex varies between $0$ and $1$,  taking  the largest values at vertices that are poorly connected to the rest of the graph \cite{puy2016random}.
On the other hand, the subspace $\lpDspm{}$ is a linear combination of filtered delta signals corresponding to the vertices in $\Smiter{}$. A filtered delta signal at a vertex $v$ is expected to decay as a function of distance from $v$. Therefore, for a particular energy $\norm{\mathbf{d}_v}^2$, a vertex whose overlap is minimum with the filtered delta signals corresponding to vertices in $\Smiter{}$ will have a small $\norm{\mathbf{P}_{\lpDspm{}} \mathbf{d}_v}^2$. A filtered delta signal $\mathbf{d}_v$ for a vertex which is farther apart from the sampled vertices will have lesser overlap with the filtered delta signals corresponding to the sampled vertices, which also span the space $\lpDspm{}$. Therefore for a vertex $v \in \SmCiter{}$ whose ``distance'' to the vertices $\Smiter{}$ is large, the corresponding filtered delta signal $\mathbf{d}_v$ will have a small projection on the space $\lpDspm{}$.

\begin{algorithm}[t]
\caption{Distance-coherence (\AlgDC)}
\begin{algorithmic}[1]
\Function{DC}{$\mathbf{L}, s, \bwidth{}, d, \epsilon$}
\State $\mathcal{S}\gets \emptyset$, \State $\Delta\gets 0.9$
\State $\mathcal{R}\gets \{1,\cdots,\bwidth{}\}$
\State $\left[\norm{\mathbf{d}_1}^2, \cdots, \norm{\mathbf{d}_n}^2\right], \lambda_f, \text{coeffs} \gets$ \Call{Compute coherence}{$\mathbf{L}, n, f, \epsilon$}
\While{$\abs{\mathcal{S}}<s$}
\State $\mathcal{V}_d(\mathcal{S}) \gets \{v \in \mathcal{S}^c| d(\mathcal{S},v) > \Delta.\max_{u\in {\mathcal{V}}} d(\mathcal{S},u)\}$
\State $v* \gets \argmaxA_{v\in \mathcal{V}_d(\mathcal{S})} \norm{\mathbf{d}_v}^2$
\State $\mathcal{S}\gets \mathcal{S}\cup v^*$
\EndWhile
\State \textbf{return} $\mathcal{S}$
\EndFunction
\end{algorithmic}
\label{alg:cohDist}
\end{algorithm}

Our proposed algorithm (see Algorithm~\ref{alg:cohDist}) consists of two stages; it first identifies vertices that are at a sufficiently large distance from already chosen vertices in $\Smiter{}$. This helps in reducing the set size, by including only those $v \in \mathcal{V}_d(\Smiter{})$ that are expected to have a small  $\norm{\mathbf{P}_{\lpDspm{}} \mathbf{d}_v}^2$. From among those selected vertices it chooses the one with the largest value of $\norm{\mathbf{d}_v}^2$.

The nodes with sufficient large distance from $S$ are defined as follows
\[
\mathcal{V}_d(\Smiter{}) = \{v \in \SmCiter{}| d(\Smiter{},v) > \Delta.\max_{u\in\mathcal{V}} d(\Smiter{},u)\},
\]
where $\Delta \in [0\; 1]$, $d(\Smiter{},v)= \min_{u \in \Smiter{}} d(u,v)$ and $d$ is the geodesic distance on the graph. The distance between two adjacent vertices $i,j$ is given by $d(i,j)=1/w(i,j)$.

The parameter $\Delta$ is used to control how many nodes can be  included in $\mathcal{V}_d(\Smiter{})$. With a small $\Delta$, more nodes will be considered at the cost of increased computations; while with a large $\Delta$, lesser nodes will be considered with the benefit of reduced computations. For small $\Delta$, the \AlgDC{} algorithm becomes similar to \AlgRand{}, except the vertices are picked in the order of their squared coherence, rather than randomly with probability proportional to their squared coherence as in \cite{puy2016random}.

The \AlgDC{} (Algorithm \ref{alg:cohDist}) provides a proof-of-concept of the volume maximization interpretation using coherences and distances for sampling. However, it involves obtaining geodesic distances on the graph, which is a computationally expensive task. Eliminating this bottleneck is possible by employing simpler distances such as hop distance, or doing away with distances altogether. We leave the first approach open for future work, and work on the second approach here as Algorithm \ref{alg:avm} (\AlgAVM{}).

\subsection{Approximate volume maximization (\AlgAVM) through inner products }
\label{sec:innProd}

In this section, we use a more efficient technique based on filtering, instead of computing the distance between nodes as in \AlgDC{},
where we assumed that the bandwidth of the signal for sampling was the same as the reconstruction bandwidth $f$. In practice, we do not know the signal bandwidth and thus also do not know the reconstruction bandwidth. To remedy this, in \AlgAVM{} we use the number of samples, $s$, as a proxy for the bandwidth of the signal. As a result, the bandwidth support used for sampling is $\mathcal{R}=\{1,\cdots,s\}$. We explained the reason behind this decoupling of the sampling and the reconstruction bandwidth in Section \ref{sec:probForm} through equations \eqref{eq:mainObj} and \eqref{eq:FreplaceR}.

\AlgAVM{} has the following advantages:
\begin{itemize}
    \item We can use the optimization framework we defined in Section \ref{sec:propAlg}.
    \item By not assuming knowledge of the reconstruction bandwidth for sampling, \AlgAVM{} models real world sampling scenarios better.
    \item For our chosen set of samples, we do not have to limit ourselves to one reconstruction bandwidth.
\end{itemize}

\begin{algorithm}[t]
\centering
\caption{Approximate volume maximization (\AlgAVM)}
\begin{algorithmic}
\Function{\AlgAVM}{$\mathbf{L},s, d, \epsilon$}
\State $\mathcal{S}\gets \emptyset$
\State $\mathcal{R}\gets \{1,\cdots,s\}$
\State $\left[\norm{\mathbf{d}_1}^2, \cdots, \norm{\mathbf{d}_n}^2\right], \lambda_s, \text{coeffs} \gets$ \Call{Compute coherence}{$\mathbf{L}, n, s, \epsilon$}
\While{$\abs{\mathcal{S}}<s$}
\State $v* \gets \argmaxA_{v\in \mathcal{S}^c} \norm{\mathbf{d}_v}^2 - \sum_{w\in \mathcal{S}}\frac{\mathbf{d}_w^2(v)}{\norm{\mathbf{d}_w}^2}$
\State $\mathbf{d}_{v^*} \gets$ \Call{Filter}{$\mathbf{L}, \text{coeffs}, \boldsymbol{\delta}_{v^*}$}
\State $\mathcal{S}\gets \mathcal{S}\cup v^*$
\EndWhile
\State \textbf{return} $\mathcal{S}$
\EndFunction
\end{algorithmic}
\label{alg:avm}
\end{algorithm}

\AlgAVM{} successively simplifies the greedy volume maximization step \eqref{eq:objMultUpdate}  in three stages.

\begin{function*}[t]
\caption{Compute coherence \cite{puy2016random}}
\begin{algorithmic}[1]
\Function{Compute coherence}{$\mathbf{L}, n, k, \epsilon$}
\State $L \gets \mathrm{round}(10\log(n))$
\State $\left[\mathbf{r}^1, \cdots, \mathbf{r}^L\right] \gets \left[\mathcal{N}(\mathbf{0}_{n\times 1}, \mathbf{I}_{n\times n}),\cdots, \mathcal{N}(\mathbf{0}_{n\times 1}, \mathbf{I}_{n\times n})\right]$
\State $\lambda_n \gets$  \Call{Approximate Largest Eigenvalue}{$\mathbf{L}$}
\State $\underline{\lambda} \gets 0$, $\bar{\lambda} \gets \lambda_n$, $\lambda \gets \lambda_n/2$.
\State $\text{coeffs} \gets$ \Call{Polynomial Filter Coefficients}{$0, \lambda_n, \lambda, d$}
\State $\left[\mathbf{r}_{\text{filt}}^1, \cdots, \mathbf{r}_{\text{filt}}^L\right] \gets \left[ \Call{Polynomial Filter}{\mathbf{L}, \text{coeffs}, \mathbf{r}^1}, \cdots, \Call{Polynomial Filter}{\mathbf{L}, \text{coeffs}, \mathbf{r}^L}\right]$
\State $SS \gets \sum_{i=1}^n \sum_{l=1}^L \; (\mathbf{r}_{\text{filt}}^l)_i^2$
\While{$\mathrm{round} \left(SS \right) \neq k$ or $\abs{\underline{\lambda}-\bar{\lambda}}>\varepsilon \cdot \bar{\lambda}$}
\label{step:start}
	\If {$\mathrm{round} \left(SS \right) \geq k$}
		\State $\bar{\lambda} \gets \lambda$.
	\Else
		\State $\underline{\lambda} \gets \lambda$.
\EndIf
\State $\lambda \gets (\underline{\lambda}+\bar{\lambda})/2$
\State $\text{coeffs} \gets$ \Call{Polynomial Filter Coefficients}{$0, \lambda_n, \lambda, d$}
\State $\left[\mathbf{r}_{\text{filt}}^1, \cdots, \mathbf{r}_{\text{filt}}^L\right] \gets \left[ \Call{Polynomial Filter}{\mathbf{L}, \text{coeffs}, \mathbf{r}^1}, \cdots, \Call{Polynomial Filter}{\mathbf{L}, \text{coeffs}, \mathbf{r}^L}\right]$
\State $SS \gets \sum_{i=1}^n \sum_{l=1}^L \; (\mathbf{r}_{\text{filt}}^l)_i^2$
\EndWhile \label{step:end}
\State $\left[\norm{\mathbf{d}_1}^2, \cdots, \norm{\mathbf{d}_n}^2\right] \gets \left[\left(\sum_{l=1}^L \; (\mathbf{r}_{\text{filt}}^l)_1^2 \right), \cdots, \left(\sum_{l=1}^L \; (\mathbf{r}_{\text{filt}}^l)_n^2 \right)\right]/SS$
\State \textbf{return} $\left[\norm{\mathbf{d}_1}^2, \cdots, \norm{\mathbf{d}_n}^2\right], \lambda, \text{coeffs}$
\EndFunction
\end{algorithmic}
\label{alg:cumCoh}
\end{function*}
\subsubsection{Approximate squared coherence}
\label{sec:approxRandCoh}
Algorithm \ref{alg:avm} estimates the squared coherence, $\norm{\mathbf{d}_v}^2, v\in\Smiter{}$, using the method of random projections method from Section 4.1 in  \cite{puy2016random} in the same way as in Algorithm  \ref{alg:cohDist}. This approach avoids explicitly finding $\mathbf{d}_v$ to compute $\norm{\mathbf{d}_v}^2$.

For completeness, we include the approach from \cite{puy2016random} to find squared coherences as Function \ref{alg:cumCoh}. For implementations of \textsc{Approximate Largest Eigenvalue}, \textsc{Polynomial Filter Coefficients}, and \textsc{Polynomial Filter} that Function \ref{alg:cumCoh} calls, we refer the reader to GSP toolbox \cite{perraudin2016gspbox}.

\subsubsection{Approximate inner product matrix}
We know that the volume of parallelepiped formed by two fixed length vectors is maximized when the vectors are orthogonal to each other. Now, since vectors that optimize  \eqref{eq:objMultUpdate} also approximately maximize the volume, we expect them to be close to orthogonal. Thus, we approximate $\lpDmam^\transp \lpDmam{}$ by an orthogonal matrix (Appendix \ref{sec:approxGramDiag}). That is, assuming that the filtered delta signals corresponding to the previously selected vertices are approximately orthogonal
we can write:
\begin{align*}
    \lpDmam^\transp \lpDmam{} &\approx \diag\roundB{\norm{\mathbf{d}_1}^2,\cdots,\norm{\mathbf{d}_m}^2},\\
    (\lpDmam^\transp \lpDmam{})^{-1} &\approx \diag\roundB{\frac{1}{\norm{\mathbf{d}_1}^2},\cdots,\frac{1}{\norm{\mathbf{d}_m}^2}},
\end{align*}
which leads to an approximation of the determinant:
\begin{align*}
    &\det\roundB{\begin{bmatrix} \lpDmam^\transp \lpDmam{} & \lpDmam^\transp \mathbf{d}_v\\ \mathbf{d}_v^\transp \lpDmam{} & \mathbf{d}_v^\transp \mathbf{d}_v\end{bmatrix}} \\&\approx \det(\lpDmam^\transp \lpDmam{}) \det(\mathbf{d}_v^\transp \mathbf{d}_v - \mathbf{d}_v^\transp\lpHatDmam{}\lpHatDmam^\transp \mathbf{d}_v),\numberthis \label{eq:detDiag}
\end{align*}
where $\lpHatDmam{}$ is obtained from $\lpDmam{}$ by normalizing the columns: $\lpHatDmam{} = \lpDmam{}\diag\roundB{1/{\norm{\mathbf{d}_1}},\cdots,1/{\norm{\mathbf{d}_m}}}$.
The second term in \eqref{eq:detDiag} can be written as:
\begin{equation}
    \mathbf{d}_v^\transp\lpHatDmam{}\lpHatDmam^\transp \mathbf{d}_v = \frac{\inner{\mathbf{d}_v}{\mathbf{d}_1}^2}{\norm{\mathbf{d}_1}^2} +
    \cdots +
     \frac{\inner{\mathbf{d}_v}{\mathbf{d}_m}^2}{\norm{\mathbf{d}_m}^2},
     \label{eq:squaredProj}
\end{equation}
which would be the signal energy of projected signal $\mathbf{d}_v$ on to  $\spn(\mathbf{d}_1,\cdots,\mathbf{d}_m)$, if the vectors $\mathbf{d}_1,\cdots,\mathbf{d}_m$ were mutually orthogonal. Note that it is consistent with our assumption that $\lpDmam^\transp\lpDmam{}$ is diagonal which would only hold if the vectors form an orthogonal set.

\subsubsection{Computing low pass filtered delta signals}
If $\mathbf{U}$ is known, then computing the low pass filtered delta signal $\mathbf{d}_v$ is straightforward with the ideal low pass filter using \eqref{eq:lpdel}. However, since we avoid eigendecomposing the Laplacian, $\mathbf{U}$ is unknown. A polynomial approximation of the ideal low pass filter with the frequency $\lambda_s$ can be computed using Function \ref{alg:cumCoh}. Using this polynomial approximation, $\boldsymbol{\delta}_v$ is filtered to obtain $\mathbf{d}_v$.

\subsubsection{Fast inner product computations}
Maximization of \eqref{eq:detDiag} requires evaluating the inner products $\inner{\mathbf{d}_v}{\mathbf{d}_i}$ in \eqref{eq:squaredProj} for all $i\in \Smiter{}$ and all vertices $v$ outside $\Smiter{}$. Suppose we knew $\mathbf{d}_i$ for sampled vertices, $i\in \Smiter{}$, and the inner products from the past iteration. The current $(m+1)^\text{th}$ iteration would still need to compute $n-m$ new inner products.

To avoid this computation we use the inner product property of  \eqref{eq:sideEffect}, which allows us to simplify  \eqref{eq:squaredProj} as follows:
\begin{equation*}
    \mathbf{d}_v^\transp\lpHatDmam{}\lpHatDmam^\transp \mathbf{d}_v = \frac{\mathbf{d}_1^2(v)}{\norm{\mathbf{d}_1}^2} +
    \cdots +
     \frac{\mathbf{d}_m^2(v)}{\norm{\mathbf{d}_m}^2}.
\end{equation*}
Thus, there is no need to compute $n-m$ new inner products, while we also avoid computing $\mathbf{d}_v, v \in \SmCiter{}$.
With this, our greedy optimization step becomes:
\begin{equation*}
    v* \gets \argmaxA_{v\in \SmCiter{}} \norm{\mathbf{d}_v}^2 - \sum_{w\in \Smiter{}}\frac{\mathbf{d}_w^2(v)}{\norm{\mathbf{d}_w}^2}.
\end{equation*}

In summary, by virtue of these series of approximations, we do not need to compute distances and no longer rely on the choice of a parameter $\Delta$, as in Algorithm \ref{alg:cohDist}. Algorithm \ref{alg:avm} only requires the following inputs:
\begin{enumerate}
    \item The number of samples requested $s$,
    \item The constant $c$ specifying $cs\log s$ random projections,
    \item The scalar $\epsilon$ specifying the convergence criteria for random projection iterations while computing squared coherences.
\end{enumerate}
The last two inputs are specifically needed by  Algorithm 1 in \cite{puy2016random}, which we use in  Step 1 (Section \ref{sec:approxRandCoh}) to compute squared coherences.

While the inner product property is defined based on the assumption that we use an ``ideal'' low pass filter for reconstruction, it can also be used to maximize the volume formed by the samples of more generic kernels --- see  Appendix \ref{app:genKern}. The approximations that we proposed in this section towards designing \AlgAVM{} can be justified if they lead to a scalable and fast algorithm. In what follows we study the computational complexity of \AlgAVM{} to assess its scalability.

\subsection{Computational complexity of \AlgAVM{}}
\label{sec:theoryComplexity}
The computational complexity of \AlgAVM{} depends on the number of vertices and edges in the graph --- $\abs{\mathcal{V}}$ and $\abs{\mathcal{E}}$, the degree of the polynomial $d$, the number of samples $s$, and the number of iterations $T_1$ to converge to the right $\lambda_s$ for computing squared coherences. In practice, we observe that a finite number of iterations $T_1$ are required to converge.

\AlgAVM{} starts by computing squared coherences, with complexity  $O(\abs{\mathcal{E}}dT_1\log\abs{\mathcal{V}})$. Finding and normalizing filtered signal requires $O(d(\abs{\mathcal{E}}+\abs{\mathcal{V}}))$ computations. Subtraction and finding the maximum requires $O(\abs{\mathcal{V}})$ computations. We repeat this $s$ times which results in  $O(s\abs{\mathcal{V}}+s(\abs{\mathcal{E}}+\abs{\mathcal{V}})d)$ computations in Stage 2 of \AlgAVM{}. This leads to Algorithm \ref{alg:avm} having a computational complexity of $O((\abs{\mathcal{E}}+\abs{\mathcal{V}})dT_1\log\abs{\mathcal{V}}+s(\abs{\mathcal{E}}+\abs{\mathcal{V}})d)$. For a connected graph we know that $\abs{\mathcal{E}}\geq \abs{\mathcal{V}} - 1$, so then the complexity is simply $O(\abs{\mathcal{E}}dT_1\log\abs{\mathcal{V}}+s\abs{\mathcal{E}}d)$.

\subsubsection{Dependence on coherence estimation accuracy}
Stage 1 is the bottleneck in the \AlgAVM{} algorithm, because it involves $T_1$ iterations to find the squared coherences, with computations in each iteration scaling as $\abs{\mathcal{E}}\log\abs{\mathcal{V}}$, where both the factors $\abs{\mathcal{E}}$ and $\abs{\mathcal{V}}$ scale with the graph size. A limitation in the number of computations we can do at this stage may cap the graph sizes we can consider.  In this situation, we note that Stage 1 (computing squared coherences and $\lambda_s$)  is an approximation, and we could select an alternative approximation requiring fewer computations instead.

\subsubsection{Dependence on the number of samples}
\label{sec:complexNSamp}
The complexity of sampling algorithms naturally depends on the number of samples requested at the input, and it is reasonable to assume that an ideal sampling algorithm cannot grow sublinearly in complexity with respect to an increase in the number of samples. This is because simply adding one sample requires $O(1)$ computations. While a sampling algorithm's complexity may grow superlinearly with the number of samples requested --- see Table III in \cite{anis2015efficient}, algorithms we compare in Section \ref{sec:avm_speed} grow linearly with respect to the requested number of samples. Additionally, \AlgAVM{}'s complexity also scales linearly with respect to an increase in the number of samples as the complexity factor $O(s\abs{\mathcal{E}}d)$ suggests.

\subsubsection{Log-linear dependence on graph size}
The computational complexity of $O(\abs{\mathcal{E}}dT_1\log\abs{\mathcal{V}}+s\abs{\mathcal{E}}d)$ suggests that \AlgAVM{} has a log-linear dependence on the graph size, specifically with a linear dependence on the number of edges and log dependence on the number of vertices. This can further be seen as complexity with log-linear dependence on the number of edges as $O(\abs{\mathcal{E}}dT_1\log\abs{\mathcal{E}}+s\abs{\mathcal{E}}d)$, but $O(\abs{\mathcal{E}}dT_1\log\abs{\mathcal{V}}+s\abs{\mathcal{E}}d)$ is a more accurate estimate.

So far, approximations to the volume maximization objective \eqref{eq:projInterpret} were useful to develop \AlgDC{} and \AlgAVM{}\footnote{Code: \url{https://github.com/STAC-USC/Graph-signal-sampling-AVM}} algorithms. In the following sections, we will show how other eigendecomposition-free algorithms can also be interpreted as approximations
to the greedy volume maximization objective.

\section{Volume maximization related work}
\label{sec:cutoff}

We next study how existing graph signal sampling methods are related to volume maximization. We start by focusing on the  \AlgSP{} algorithm from \cite{anis2015efficient} and show how can it be seen as a volume maximization method. This idea is developed in Section \ref{sec:specDesc} and Section \ref{sec:volParal}. Section \ref{sec:tradeoffInn} then considers other eigendecomposition-free methods and draw parallels with our volume maximization approach.

\subsection{\AlgSP{} algorithm as Gaussian elimination}
\label{sec:specDesc}

The \AlgSP{} algorithm is based on the following theorem.
\begin{theorem}
\cite{anis2015efficient} Let $\mathbf{L}$ be the combinatorial Laplacian of an undirected graph. For a set $\Smiter{}$ of size $m$, let $\mathbf{U}_{\Smiter{},1:m}$ be full rank. Let $\boldsymbol{\psi}^*_k$ be zero over $\Smiter{}$ and a minimizing signal in the Rayleigh quotient of $L^k$ for a positive integer $k$.
\begin{equation}
    \boldsymbol{\psi}^*_k = \argminA_{\boldsymbol{\psi},\boldsymbol{\psi}(\Smiter{})=\mathbf{0}} \frac{\boldsymbol{\psi}^\transp \mathbf{L}^k \boldsymbol{\psi}}{\boldsymbol{\psi}^\transp \boldsymbol{\psi}}.
    \label{eq:rayOpt}
\end{equation}
Let the signal $\boldsymbol{\psi}^*$ be a linear combination of first $m+1$ eigenvectors such that $\boldsymbol{\psi}^*(\Smiter{})=\mathbf{0}$.  Now if there is a gap between the singular values $\sigma_{m+2}>\sigma_{m+1}$, then  $\norm{\boldsymbol{\psi}^*_k- \boldsymbol{\psi}^*}_2\rightarrow 0$ as  $k\rightarrow \infty$.
\label{thm:eigConverge}
\end{theorem}
\begin{proof}
\cite{anis2015efficient}, for $l_2$ convergence see Appendix \ref{app:conv}.
\end{proof}
Following \eqref{eq:rayOpt}, the next step in the \AlgSP{} algorithm that leads to sampling a new vertex is
\begin{equation*}
    v^* = \argmaxA_{v\in \Smiter{}^c} \abs{\boldsymbol{\psi}^*_k},
\end{equation*}
where $\boldsymbol{\psi}^*_k$ is from \eqref{eq:rayOpt}.
Consider first the ideal \AlgSP{} algorithm, where $k\rightarrow\infty$ and the solution tends to the ideal bandlimited solution.
In the ideal case, given a full rank $\mathbf{U}_{\Smiter{},1:m}$, from  Theorem \ref{thm:eigConverge} we can always get another vertex $v$ such that $\mathbf{U}_{\Smiter{}\cup v,1:m+1}$ is also full rank. Thus, at all iterations any submatrix of  $\mathbf{U}_{\Smiter{},1:m}$ will have full rank.

When $k\rightarrow\infty$ and $\Smiter{}$ vertices are selected, $\boldsymbol{\psi}^*$ is given by the $m+1^\text{th}$ column of $\mathbf{U'}$, where $\mathbf{U}'$ is obtained by applying Gaussian elimination to the columns of $\mathbf{U}$ such that the $m+1^\text{th}$ column has zeros at indexes given by $\Smiter{}$ \cite{anis2015efficient}. $\mathbf{U}'$ can be written as:
\begin{align}
    &\mathbf{U}' = \nonumber \\ &\begin{bmatrix}
    \boldsymbol{u}_1(1) & & \multicolumn{2}{c}{\multirow{2}{*}{\Large{$\mathbf{0}$}}} & & & \\
    & \boldsymbol{u}^{'}_2(2) & \multicolumn{2}{c}{} & \vrule & & \vrule\\
    & & \ddots & & \mathbf{u}_{m+2} & \cdots &  \mathbf{u}_n\\
    & & & \boldsymbol{u}^{'}_{m+1}(m+1) &  &  &  \\
    \multicolumn{2}{c}{\multirow{2}{*}{\Large{$\bigstar$}}} & & \vdots & \vrule & & \vrule\\
    \multicolumn{2}{c}{} & & \boldsymbol{u}^{'}_{m+1}(n) &  & &
    \end{bmatrix},
    \label{eq:elimMat}
\end{align}
where $\star$ denotes arbitrary entries and $\mathbf{0}$ denotes zero entries in the corresponding matrix regions. Because we have non-zero pivots, $\boldsymbol{u}_1(1)$ to $\boldsymbol{u}^{'}_{m+1}(m+1)$, $\mathbf{U}_{\Smiter{},1:m}$ is full rank. The columns in $\mathbf{U}$ from $m+2$ to $n$ remain intact.

The next section explains how an iteration in the ideal \AlgSP{} algorithm can be seen as a volume maximization step.

\begin{figure*}[t]
\centering
\tdplotsetmaincoords{70}{30}
\subfloat[Orthogonality of $\mathbf{h}$ with $\lpTildeDspm{}$.]{
  \begin{tikzpicture}[tdplot_main_coords, 2dspace/.style={fill=black!50!white, opacity=0.3},scale=0.48]
        \tikzstyle{every node}=[font=\small]
\draw[thick,-{Stealth[length=2mm,width=2mm]}] (0,0,0) -- (5,5,0) node[anchor=north west]{$\mathbf{d}_1$};
\draw[thick,-{Stealth[length=2mm,width=2mm]}] (0,0,0) -- (3.53,6.12,0) node[anchor=south west]{$\mathbf{d}_2$};
\draw[thick,-{Stealth[length=2mm,width=2mm]}] (0,0,0) -- (-5,5,0) node[anchor=south]{$\mathbf{d}_m$};
\draw[thick,-{Stealth[length=2mm,width=2mm]}] (0,0,0) -- (0,0,5) node[anchor=south]{$\mathbf{h}$};

\filldraw[2dspace] (7,7,0) -- (-7,7,0) -- (-7,-1,0) node[below right, text=black,opacity=1]{$\lpTildeDspm{}$}-- (7,-1,0) -- (7,7,0) ;
\end{tikzpicture}
\label{fig:hOrtho}
}
\subfloat[Components of $\mathbf{d}_v$.]{
  \begin{tikzpicture}[tdplot_main_coords, 2dspace/.style={fill=black!50!white, opacity=0.3},scale=0.48]
      \tikzstyle{every node}=[font=\small]
      \draw[thick,-{Stealth[length=2mm,width=2mm]}] (0,0,0) -- (1.8,5.2,4.8) node[anchor=south ]{$\mathbf{d}_v$};
      \draw[thick,-{Stealth[length=2mm,width=2mm]}] (0,0,0) -- (1.8,5.2,0) node[below,yshift=-0.5cm]{$\mathbf{P}_{\lpTildeDspm{}}(\mathbf{d}_v)$};

\filldraw[2dspace] (7,7,0) -- (-7,7,0) -- (-7,-1,0)node[below right,text =black,opacity=1]{$\lpTildeDspm{}$} -- (7,-1,0) -- (7,7,0);
\draw[thick,dashed] (1.8,5.2,4.8)  -- (1.8,5.2,0) node[above right, yshift=0.5cm]{$\abs{\inner{\mathbf{h}}{\mathbf{d}_v}}$};

\fill[black] (1.8,5.2,0) circle (2pt);
\end{tikzpicture}
\label{fig:dvProj}
}
\caption{Geometry of \AlgSP{}.}
\end{figure*}
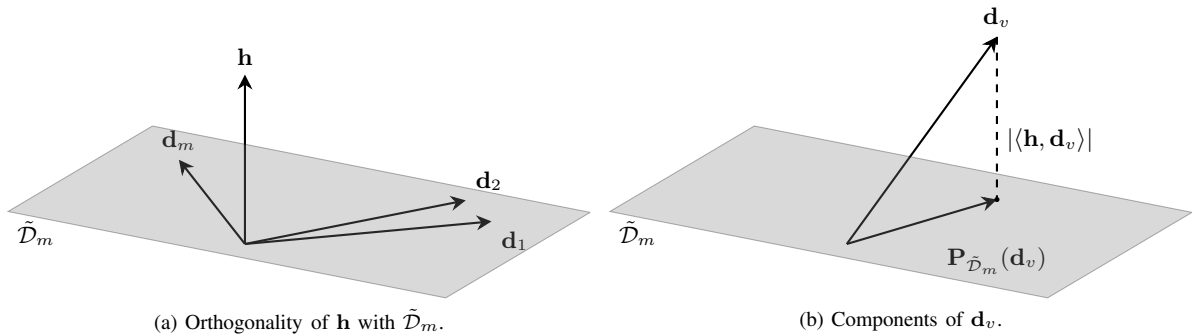

\subsection{\AlgSP{} algorithm as volume maximization}
\label{sec:volParal}
Consider a single stage in the \AlgSP{} algorithm, where  the current sampling set is $\Smiter{}$,  $\abs{\Smiter{}}=m$ is the number of sampled vertices, and the bandwidth support is $\bwfm{} = \{1,\cdots,m+1\}$. At this stage, choosing a vertex $v$ that maximizes $\det(\mathbf{U}_{\Smiter{}\cup v,\bwfm{}}\mathbf{U}_{\Smiter{}\cup v,\bwfm{}}^\transp)$ is equivalent to choosing a vertex that maximizes $\abs{\det(\mathbf{U}_{\Smiter{}\cup v, \bwfm{}})}$. We briefly state our results in terms of $\abs{\det(\mathbf{U}_{\Smiter{}\cup v, \bwfm{}})}$ as this gives us the additional advantage of making the connection with the Gaussian elimination concept we discussed in Section \ref{sec:specDesc}. Now, focusing on the selection of the $(m+1)^\text{th}$ sample, we can state the following result.

\begin{prop}
The sample $v^*$ selected in the $(m+1)^\text{th}$ iteration of the ideal \AlgSP{} algorithm is the vertex $v$ from $\SmCiter{}$ that maximizes $\abs{\det(\mathbf{U}_{\Smiter{}\cup v, \bwfm{}})}$. \label{prop:detPivot}
\end{prop}
\begin{proof}
The ideal \AlgSP{} algorithm selects the vertex corresponding to the maximum value in $\abs{\boldsymbol{u}^{'}_{m+1}(m+1)}$, $\cdots$, $\abs{\boldsymbol{u}^{'}_{m+1}(n)}$. Since $\Smiter{}$ is given and $\mathbf{U}'_{\Smiter{}\cup v, \bwfm{}}$ is a diagonal matrix, this also corresponds to selection of $v$ such that magnitude value of the  $\det(\mathbf{U}_{\Smiter{}\cup v, \bwfm{}}')$ is the maximum among all possible $v$ selections.

But because $\mathbf{U}'_{\Smiter{}\cup v, \bwfm{}}$ is obtained from $\mathbf{U}$ by doing Gaussian elimination, the two determinants are equal, i.e.,   $\abs{\det(\mathbf{U}_{\Smiter{}\cup v, \bwfm{}})} = \abs{\det(\mathbf{U}_{\Smiter{}\cup v, \bwfm{}}')}$, and since the current $(m+1)^\text{th}$ iteration chooses the maximum absolute value pivot, given $\Smiter{}$ that sample maximizes $\abs{\det(\mathbf{U}_{\Smiter{}\cup v, \bwfm{}})}$.
\end{proof}

We now show that the vertex $v^*$ is selected in the $(m+1)^\text{th}$ iteration according to the following rule:
\begin{equation*}
    v^* = \argmaxA_{v\in \SmCiter{}} \text{dist}(\mathbf{d}_v, \spn(\mathbf{d}_1,\cdots,\mathbf{d}_m)),
\end{equation*}
where $\text{dist}(\cdot, \cdot)$ is the distance between a  vector and its orthogonal projection onto a vector subspace. Thus, this optimization is equivalent to selecting a vertex $v$ that maximizes the volume of the contained parallelepiped, $\text{Vol}(\mathbf{d}_1,\cdots,\mathbf{d}_m, \mathbf{d}_v)$.

Let $\mathbf{h}$ be a unit vector along the direction of  $(m+1)^\text{th}$ column of $\mathbf{U}'$ in  \eqref{eq:elimMat}. We are interested in finding the vertex $v$ that maximizes $\abs{\mathbf{h}(v)}$.

\begin{prop}
The signal value $\mathbf{h}(v)$ is the length of projection of $\mathbf{d}_v$ on $\mathbf{h}$.
\label{prop:indexProj}
\end{prop}
\begin{proof}
The signal $\mathbf{h}$ belongs to the bandlimited space, $\mathbf{h}\in \spn(\mathbf{U}_{\bwfm{}})$. Thus, using \eqref{eq:rkDot} we have that:
\begin{equation*}
    \mathbf{h}(v) = \inner{\mathbf{h}}{\mathbf{d}_v}.
\end{equation*}
Since $\mathbf{h}$ is a unit vector, the last expression in the equation above is the projection length of $\mathbf{d}_v$ on $\mathbf{h}$.
\end{proof}

So $\abs{\mathbf{h}(v)}$ is maximized when $\abs{\inner{\mathbf{d}_v}{\mathbf{h}}}$ is maximized.

\begin{prop}
The signal $\mathbf{h}$ is such that $\mathbf{h}\in \spn(\mathbf{d}_1,\allowbreak\cdots,\mathbf{d}_m)^{\perp} \cap \spn(\mathbf{U}_{\bwfm{}})$.
\label{prop:orthComp}
\end{prop}
\begin{proof}
All pivots in the Gaussian elimination of $\mathbf{U}_{\Smiter{}\cup v, \bwfm{}}$ are non-zero, as seen in \eqref{eq:elimMat}, so that the following equivalent statements follow:
\begin{itemize}
    \item $\mathbf{U}_{\Smiter{},1:m}$ is full rank.
    \item $\mathbf{U}_{\bwfm{}}\mathbf{U}_{\bwfm{}}^\transp \mathbf{I}_{\Smiter{}}$ is full column rank.
    \item $\spn(\mathbf{d}_1,\cdots,\mathbf{d}_m)$ has dimension $m$.
\end{itemize}
The second statement follows (0.4.6 (b) \cite{horn2012matrix}) from the first because $\mathbf{U}_{\bwfm{}}$ has full column rank and $\mathbf{U}_{\Smiter{},1:m}$ is nonsingular. Given that  $\spn(\mathbf{d}_1,\cdots,\mathbf{d}_m)$ has dimension $m$ we can proceed to the orthogonality arguments.

By definition, $\mathbf{h}$ obtained from \eqref{eq:elimMat} is zero over the set $\Smiter{}$ so that, from Proposition \ref{prop:rkhsDot}:
\begin{align*}
    \mathbf{h}(1) = 0 &\implies \inner{\mathbf{h}}{\mathbf{d}_1} = 0,\\
    &\vdots\\
    \mathbf{h}(m) = 0 &\implies \inner{\mathbf{h}}{\mathbf{d}_m} = 0,
\end{align*}
and therefore $\mathbf{h}$ is orthogonal to each of the vectors $\mathbf{d}_1,\allowbreak\cdots,\mathbf{d}_m$. We call the space spanned by those vectors $\lpTildeDspm{}$, defined as
\begin{equation*}
    \lpTildeDspm{} = \{\spn(\mathbf{d}_i)| i\in \{1, 2, \cdots,m\}\}.
\end{equation*}
Since dimension of $\mathbf{U}_{\bwfm{}}$ is $m+1$ and $\mathbf{h}$ is orthogonal to $\lpTildeDspm{} = \spn(\mathbf{d}_1,\cdots,\mathbf{d}_m)$ of dimension $m$, $\spn{(\mathbf{h})}$ is the orthogonal complement subspace to $\lpTildeDspm{}$ (see Fig. \ref{fig:hOrtho} for an illustration).
\end{proof}
For this particular algorithm, $\bwf{}$ changes with the number of samples in the sampling set. At the end of $m^\text{th}$ iteration the bandwidth support $\bwf$ can be represented as $\bwfm{}=\{1,\cdots,m+1\}$, where $m$ is the number of samples in the current sampling set. We use $\lpTildeDsp{}$ and $\lpTildeDspm{}$ to denote a dependence of $\lpDsp{}$ and $\lpDspm{}$ on $\bwfm{}$ in addition to $\Smiter{}$.

\begin{prop}
The sample $v^*$ selected in the $(m+1)^\text{th}$ iteration of \AlgSP{} maximizes the distance between $\mathbf{d}_v$ and its orthogonal projection on $\lpTildeDspm{}$.
\label{prop:minOrthDist}
\end{prop}
\begin{proof}
Since $\mathbf{d}_v\in \spn(\mathbf{U}_{\bwfm{}})$, it can be resolved in to two orthogonal components with respect to the two orthogonal (Prop. \ref{prop:orthComp}) spaces $\lpTildeDspm{}$ and $\spn{(\mathbf{h})}$.
\begin{equation*}
    \mathbf{d}_v = \mathbf{P}_{\lpTildeDspm{}} \mathbf{d}_v + \inner{\mathbf{d}_v}{\mathbf{h}}\mathbf{h},
\end{equation*}
where $\mathbf{h}$ has unit magnitude and $\mathbf{P}_{\lpTildeDspm{}}$ is the projection matrix onto the subspace $\lpTildeDspm{}$.

Maximizing $\abs{\mathbf{h}(v)}$ is equivalent to maximizing $\abs{\inner{\mathbf{d}_v}{\mathbf{h}}}$ which can be expressed in terms of magnitude of $\mathbf{d}_v$ and the magnitude of its projection on $\lpTildeDspm{}$.
\begin{equation}
  \argmaxA_v \inner{\mathbf{d}_v}{\mathbf{h}}^2 = \argmaxA_v \norm{\mathbf{d}_v}^2 - \norm{\mathbf{P}_{\lpTildeDspm{}} \mathbf{d}_v}^2
  \label{eq:obj}
\end{equation}
Fig. \ref{fig:dvProj} shows this orthogonality relation between $\abs{\inner{\mathbf{h}}{\mathbf{d}_v}}$, $\mathbf{d}_v$, and $\mathbf{P}_{\lpTildeDspm{}}(\mathbf{d}_v)$.
\end{proof}

So the $v^*$ chosen is the one that maximizes the volume of the space spanned by the filtered delta signals.
\begin{equation*}
    v^* = \argmaxA_{v\in \SmCiter{}} \text{Vol}(\mathbf{d}_1,\cdots,\mathbf{d}_m,\mathbf{d}_v).
\end{equation*}
The last line follows from using the definition of volume of parallelepiped \cite{peng2007determinant}. Note that, although Proposition \ref{prop:minOrthDist} could have been derived from the determinant property in Proposition  \ref{prop:detPivot} using \eqref{eq:detEqui}, the approach using the orthogonal vector to the subspace in Proposition \ref{prop:indexProj}, Proposition \ref{prop:orthComp} and Proposition  \ref{prop:minOrthDist} makes more explicit the geometry of the problem.

Algorithm \ref{alg:spProx} summarizes this new volume maximization interpretation of \AlgSP{}. Although Algorithm \ref{alg:spProx} requires eigendecomposition, it is helpful to see its conceptual similarity with Algorithms \ref{alg:cohDist} and \ref{alg:avm}. For an empirical comparison, in Section \ref{sec:expVol} we compare \AlgSP{} which is Algorithm \ref{alg:spProx} relaxed with a finite value of $k$ and without requiring the full eigendecomposition.

\begin{algorithm}[t]
\centering
\caption{Volume interpretation of \AlgSP{} algorithm as $k\rightarrow \infty$}
\begin{algorithmic}
\Function{\AlgSP{}}{$\mathbf{L}, s$}
\State $\mathcal{S}\gets \emptyset$
\State $\bwf{}\gets \{1\}$
\While{$\abs{\mathcal{S}}<s$}
\State $\left[\mathbf{d}_1, \cdots, \mathbf{d}_{\abs{\mathcal{S}}} \right] \gets \left[\mathbf{U}_{\mathcal{R}}\mathbf{U}_{\mathcal{R}}^\transp\boldsymbol{\delta}_1, \cdots, \mathbf{U}_{\mathcal{R}}\mathbf{U}_{\mathcal{R}}^\transp\mathbf{d}_{\abs{\mathcal{S}}} \right]$
\State ${\lpTildeDsp{}} \gets \spn_{v\in\mathcal{S}}(\mathbf{d}_v)$
\State $v* \gets \argmaxA_v \norm{\mathbf{d}_v}^2 - \norm{\mathbf{P}_{{\lpTildeDsp{}}} \mathbf{d}_v}^2$
\State $\mathcal{S}\gets \mathcal{S}\cup v^*$
\State $\bwf{}\gets \bwf{}\cup \abs{\bwf{}}+1$
\EndWhile
\State \textbf{return} $\mathcal{S}$
\EndFunction
\end{algorithmic}
\label{alg:spProx}
\end{algorithm}

\begin{table*}[t]
    \centering
    \caption{Approximation to greedy maximization of determinant. \AlgEDfree{} - For implementation details refer to Section \ref{sec:tradeoffInn}}
    \begin{tabular}{| c | c | c |}\hline
        Sampling method & Selection process & Approximation \\ \hline
        Exact greedy & $\argmaxA_{v\in \SmCiter{}} \norm{\mathbf{d}_v}^2 - \norm{\mathbf{P}_{\lpDspm{}} \mathbf{d}_v}^2$ & -\\ \hline
        \AlgRand{} & $p(v) \propto \norm{\mathbf{d}_v}^2$ & No projection.\\ \hline
        \AlgSP{} & $\argmaxA_{v\in \SmCiter{}} \norm{\mathbf{d}_v}^2 - \norm{\mathbf{P}_{\lpTildeDspm{}} \mathbf{d}_v}^2$ & Projection space approximate and increasing in size.\\ \hline
        \AlgEDfree{} & $\argmaxA_{v\in \SmCiter{}} \norm{\mathbf{d}_v}^2 - 2 \sum_{w\in\Smiter{}} \inner{\abs{\mathbf{d}_w}}{\abs{\mathbf{d}_v}}$ & Inner product approximation for projection.\\ \hline
        \AlgAVM{} & $\argmaxA_{v\in \SmCiter{}} \norm{\mathbf{d}_v}^2 - \sum_{w\in \Smiter{}}\frac{\mathbf{d}_w^2(v)}{\norm{\mathbf{d}_w}^2}$ & Assumption of orthogonality.\\ \hline
    \end{tabular}
    \label{tab:approx_Comp}
\end{table*}

\subsection{Eigendecomposition-free methods as volume maximization}
\label{sec:tradeoffInn}
So far we have covered existing literature on D-optimality both in general and as it relates  to graphs, and proposed two algorithms towards that goal. From \eqref{eq:projInterpret}, note that the greedy update for approximate volume maximization is
\begin{equation*}
    v^* = \argmaxA_{v\in \SmCiter{}} \norm{\mathbf{d}_v}^2 - \norm{\mathbf{P}_{\lpDspm{}} \mathbf{d}_v}^2.
\end{equation*}
Based on this we can revisit some eigendecomposition-free algorithms for graph signal sampling from the perspective of volume maximization. For each of these algorithms, we consider the criterion to add a vertex to the sampling set in the $(m+1)^\text{th}$ iteration.
First, the \AlgRand{} algorithm can be seen as neglecting the projection term and sampling based only on $\norm{\mathbf{d}_v}^2$. Alternatively, the \AlgSP{} algorithm approximates this by
\begin{equation}
    v^* = \argmaxA_{v\in \SmCiter{}} \norm{\mathbf{d}_v}^2 - \norm{\mathbf{P}_{\lpTildeDspm{}} \mathbf{d}_v}^2
    \label{eq:spUpdate}
\end{equation}
for a finite value of parameter $k$ and a varying $\lpTildeDspm{}$ in place of $\lpDspm{}$ --- see Section \ref{sec:volParal}.
The generalization to the
eigendecomposition-free approach of \cite{sakiyama2018eigendecomposition}(V2) proposes to maximize  (using the greedy selection in Equation (31)):
\begin{equation*}
v^* = \argmaxA_{v\in \SmCiter{}} \norm{\mathbf{d}_v}^2 - 2\sum_{w\in \Smiter{}} \inner{\abs{\mathbf{d}_w}}{\abs{\mathbf{d}_v}},
\end{equation*}
but in practice maximizes a different expression --- see (32) in \cite{sakiyama2019eigendecomposition}.

The crucial difference between our proposed method and \cite{sakiyama2019eigendecomposition} is that we obtain a specific expression to be maximized through D-optimality. Whereas \cite{sakiyama2019eigendecomposition} clearly shows the relation between various experiment design objective functions and their corresponding localization operators, the relation between the algorithm proposed in \cite{sakiyama2019eigendecomposition}, \AlgEDfree{},
and the experiment design objective functions is unclear.

We do not attempt to explain methods such as \cite{jung2019localized} under the volume maximization framework as they define the signal smoothness through total variation operator as opposed to squared differences through the graph Laplacian operator which is necessary for the volume maximization interpretation. The similarities in the optimization objective function for various eigendecomposition-free sampling methods that we studied are summarized in Table \ref{tab:approx_Comp}. The differences between various sampling methods will be apparent when we compare their performance for various sampling and reconstruction settings.

\section{Experimental settings}
\label{sec:expVol}
To evaluate our sampling algorithms, we assess their sampling  performance on different graph topologies at different sampling rates.

\subsection{Signal, Graph Models and Sampling setups}
\label{sec:expSetup}
With a perfectly bandlimited signal, most sampling schemes achieve similar performance in terms of reconstruction error. However, in practice signals are rarely perfectly bandlimited and noise-free. Therefore, it is necessary to compare the performance of the sampling methods on non-ideal signals.

\subsubsection{Signal smoothness and graph topologies}\label{sec:sig_smooth} Consider first  a synthetic noisy signal model. The signal $\mathbf{f}$ is bandlimited with added noise $\mathbf{n}$. The resulting signal is $\mathbf{f} = \mathbf{x} + \mathbf{n}$ which can be expressed as $\mathbf{U}_{\recF{}}\tilde{\mathbf{f}}_\recF{} + \mathbf{n}$
with the frequency components of $\mathbf{x}$, and noise being random variables distributed as multivariate normal distributions: $\tilde{\mathbf{f}}_\recF{} \sim \mathcal{N}(\mathbf{0},c_1\mathbf{I}_{\recF{}\recF{}})$, $\mathbf{n} \sim \mathcal{N}(\mathbf{0},c_2\mathbf{I}_{\mathcal{V}\mathcal{V}})$. The constants $c_1$ and $c_2$ are chosen so that the expected signal power is $1$ and the expected noise power is $0.1$. Since our main objective is to study the effect of varying number of samples, graph topologies, and the graph size on \AlgDC{} and \AlgAVM{}, we fix the bandwidth of the signal to $50$.

We compare our algorithms against three established algorithms --- \AlgRand{}, \AlgSP{}, and \AlgEDfree{}
All methods except \AlgRand{} return unique samples. For a fair comparison, all sampling methods are evaluated under conditions where the same number of samples is obtained, irrespective of whether the returned ones are unique or not (which could occur in the case of \cite{puy2016random}).

We use the combinatorial Laplacian for our sampling and reconstruction experiments, except for the classification experiment where the normalized Laplacian of the nearest neighbors graph is used, as it achieves overall better classification accuracy.

\subsubsection{Sampling set sizes}
We use two graph sizes $500$ and $1000$. Except for the Erd\H{o}s R\'{e}nyi graph model,
we use the Grasp \cite{girault2017grasp} and GSPBox \cite{perraudin2016gspbox} \softMat{} toolboxes to generate the graph instances --- see Table \ref{tab:grModel}.
\begin{table*}[t]
    \centering
    \caption{Types of graphs in the experiments}
    \begin{tabular}{| c | c | c |}\hline
        Graph model & Instance & Construction comments \\ \hline
        Random sensor knn & \texttt{grasp\_plane\_knn(n, k),k=8 or 15} & Uniformly sampled vertices on 2d plane with $k$ nearest neighbors\\
        & \texttt{gsp\_sensor(n, 20)} & Uniformly sampled vertices on 2d plane with 20 nearest neighbors\\ \hline
        Scale-free & \texttt{grasp\_barabasi\_albert(n,8)} & Initial 8 nodes \\ \hline
        Community & \texttt{param.Nc=5 or 10} & \texttt{param.Nc} communities\\
        & \texttt{gsp\_community(n, param)} & \\ \hline
        WS/Small world &  \texttt{grasp\_watts\_strogatz(n, 5, 0.2)} & Average degree 5 rewiring probability $0.2$\\ \hline
        Erd\H{o}s R\'{e}nyi & \texttt{erdos\_renyi(n,0.02)} & Probability of connection $0.02$\\ \hline
    \end{tabular}
    \label{tab:grModel}
\end{table*}

We use sampling set sizes from $60$ samples to $200$ samples to compare the variation of reconstruction error.
For comparing algorithms in this setting, we do not show the full range of reconstruction SNRs from \AlgRand{}  because its SNR is usually 5-10 dBs lower than other methods at starting sampling rate of 60 (see Tables \ref{tab:reconScale} and  \ref{tab:reconScaleComm} for performance at higher sampling rates).

\subsubsection{Classification on real world datasets}
In this experiment we evaluate sampling algorithms in a transductive semi-supervised learning setting for a digit classification task (USPS dataset). We randomly select 10 smaller datasets of size 1000 from the original dataset, such that each smaller dataset contains 100 elements from each category, i.e., the 10 digits. Using those smaller datasets we construct a nearest neighbors graph with 10 neighbors. This setup is the same as in \cite{anis2015efficient}. The graph sampling and reconstruction algorithms then select number of samples ranging from $60$ to $200$. Using one-vs-all strategy we reconstruct the class signals, and then classify them by selecting the class which gives the maximum reconstruction in magnitude at a vertex. We then report the average accuracy of the classification over the 10 smaller sets.

\subsubsection{Effect of scaling graph sizes}
\label{sec:effectGraphSizes}
One of our primary goals is to develop fast and scalable  algorithms. To evaluate these properties,  we use a random sensor nearest neighbors graph with 20 nearest neighbors, and community graph with 10 communities with different graphs sizes of  500, 1000, 2000, 4000, 8000. For each graph size we sample 150 vertices, and the signal model remains the same as in Section \ref{sec:sig_smooth}. We report the SNR of the reconstruction and the time required to sample averaged over 50 graph and signal realizations for a given graph type.

The feasibility of different sampling algorithms on graphs with sizes orders of magnitude larger than thousands of vertices is an indicator of scalability, so we also test the sampling algorithms on relatively larger graph sizes of 50,000 and 100,000. Both the graph parameters, such as the number of nearest neighbors or the number of communities, and the signal model parameters, such as the noise power, remain the same as those we use for smaller graph sizes; while we use a bandwidth of 100, and sample 5000 vertices for the two larger graphs sizes. At these graph sizes, some sampling algorithms require more than 10 times the time required by \AlgAVM{} or require more than 64GB of the available random-access memory (RAM). Because of this, it is not possible to run 50 graph and signal realizations for all the sampling algorithms as we did earlier. Least squares reconstruction which we used for smaller graphs is also not feasible at these graph sizes, so we reconstruct using projection onto convex sets(\ProjCS{}) from \cite{narang2013localized}, which is tailored for bandlimited signals on graphs. We include the execution times and the SNRs for this setting in the tables along with smaller graphs, but due to fundamental difference in the reconstruction method we do not plot the SNRs together with those of smaller graphs. The execution times measured in seconds are rounded to one decimal precision for display.

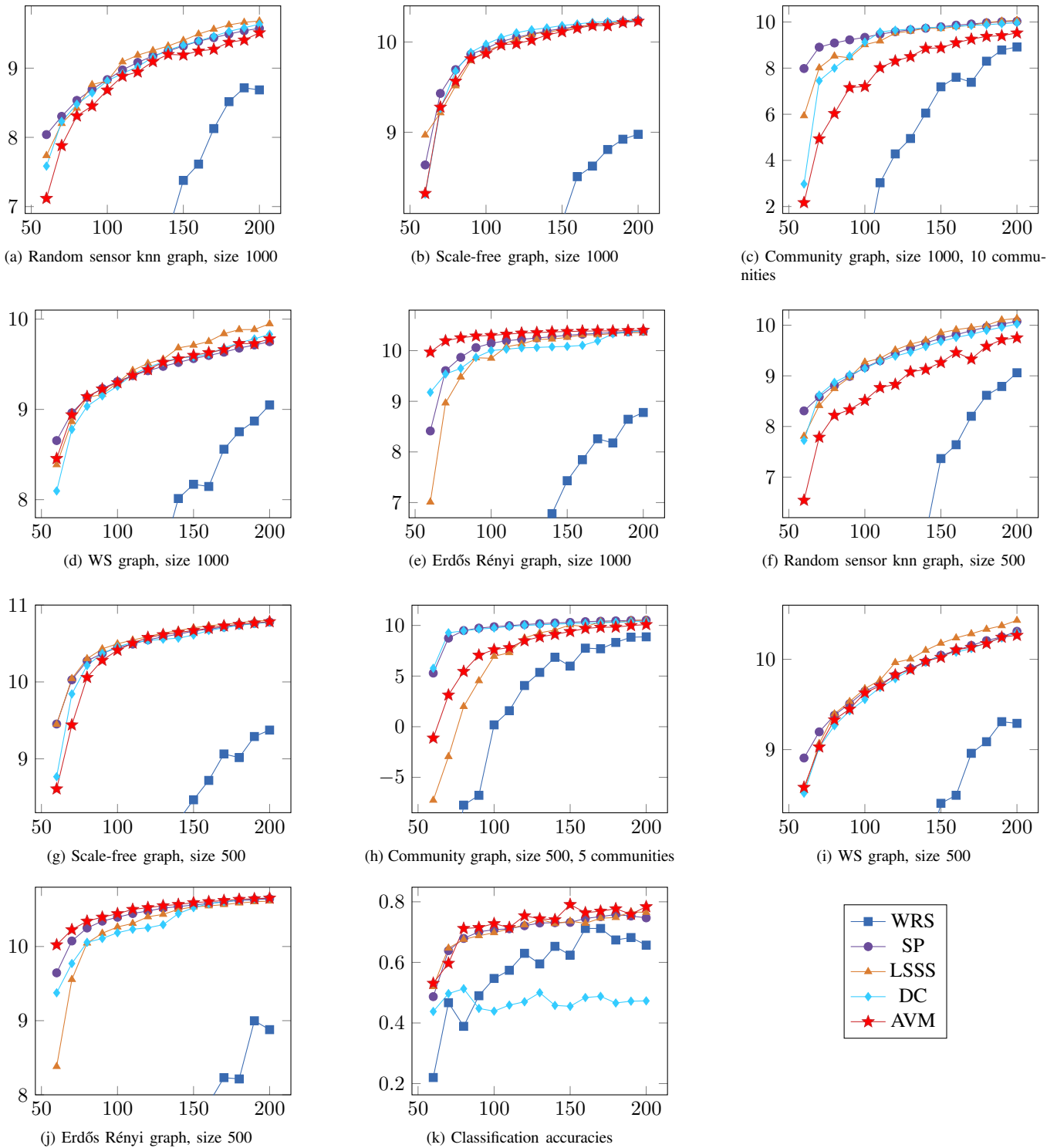
\begin{figure*}[th!]
\subfloat[Random sensor knn graph, size 1000]{
\pgfplotstableread{knn_bandlim_sig1_1000_8it50_snr.dat}\mytable
\begin{tikzpicture}
\centering
  \begin{axis}[
  	width = 0.33\linewidth,
    cycle list name=mydccol,
    ymin = 6.9,
    ymax = 9.9
    ]
    \pgfplotstablegetcolsof{\mytable}
    \pgfmathparse{\pgfplotsretval-1}
    \foreach \i in {1,2}{
      \addplot table[x index=0,y index=\i] {\mytable};
    }
    \addplot table[
    x index=0,
    y ={create col/expr={max(\thisrow{3},\thisrow{4},\thisrow{5},\thisrow{6},\thisrow{7})}}] {\mytable};
    \foreach \i in {8,9}{
      \addplot table[x index=0,y index=\i] {\mytable};
    }
  \end{axis}
  \label{fig:knn1}
\end{tikzpicture}
}\hfill
\subfloat[Scale-free graph, size 1000]{
\pgfplotstableread{BA_bandlim_sig1_1000_8it50_snr.dat}\mytable
\begin{tikzpicture}
\centering

  \begin{axis}[
  	width = 0.33\linewidth,
    cycle list name=mydccol,
    ymin=8.1,
    ymax=10.4
    ]
    \pgfplotstablegetcolsof{\mytable}
    \pgfmathparse{\pgfplotsretval-1}
    \foreach \i in {1,2}{
      \addplot table[x index=0,y index=\i] {\mytable};
    }
    \addplot table[
    x index=0,
    y ={create col/expr={max(\thisrow{3},\thisrow{4},\thisrow{5},\thisrow{6},\thisrow{7})}}] {\mytable};
    \foreach \i in {8,9}{
      \addplot table[x index=0,y index=\i] {\mytable};
    }
  \end{axis}
\end{tikzpicture}
\label{fig:ba1}
}\hfill
\subfloat[Community graph, size 1000, 10 communities]{
\pgfplotstableread{comm_bandlim_sig1_1000_10it50_2020_12_4_21_53_snr.dat}\mytable
\begin{tikzpicture}

  \begin{axis}[
  	width = 0.33\linewidth,
    cycle list name=mydccol,
    ymin=1.7,
    ymax=10.7
    ]
    \pgfplotstablegetcolsof{\mytable}
    \pgfmathparse{\pgfplotsretval-1}
    \foreach \i in {1,2}{
      \addplot table[x index=0,y index=\i] {\mytable};
    }
    \addplot table[
    x index=0,
    y ={create col/expr={max(\thisrow{3},\thisrow{4},\thisrow{5},\thisrow{6},\thisrow{7})}}] {\mytable};
    \foreach \i in {8,9}{
      \addplot table[x index=0,y index=\i] {\mytable};
    }
  \end{axis}
\end{tikzpicture}
\label{fig:comm1}
}\hfill \\
\subfloat[WS graph, size 1000]{
\pgfplotstableread{WS_bandlim_sig1_1000_5_p2it50_snr.dat}\mytable

\begin{tikzpicture}
  \begin{axis}[
  	width = 0.33\linewidth,
    cycle list name=mydccol,
    ymin=7.8,
    ymax=10.1,
    ]
    \pgfplotstablegetcolsof{\mytable}
    \pgfmathparse{\pgfplotsretval-1}
    \foreach \i in {1,2}{
      \addplot table[x index=0,y index=\i] {\mytable};
    }
    \addplot table[
    x index=0,
    y ={create col/expr={max(\thisrow{3},\thisrow{4},\thisrow{5},\thisrow{6},\thisrow{7})}}] {\mytable};
    \foreach \i in {8,9}{
      \addplot table[x index=0,y index=\i] {\mytable};
    }
  \end{axis}
\end{tikzpicture}
\label{fig:ws1}
}\hfill
\subfloat[Erd\H{o}s R\'{e}nyi graph, size 1000]{
\pgfplotstableread{ER_bandlim_sig1_1000_p02it50_snr.dat}\mytable

\begin{tikzpicture}
  \begin{axis}[
  	width = 0.33\linewidth,
    cycle list name=mydccol,
    ymin=6.7,
    ymax=10.8,
    ]
    \pgfplotstablegetcolsof{\mytable}
    \pgfmathparse{\pgfplotsretval-1}
    \foreach \i in {1,2}{
      \addplot table[x index=0,y index=\i] {\mytable};
    }
    \addplot table[
    x index=0,
    y ={create col/expr={max(\thisrow{3},\thisrow{4},\thisrow{5},\thisrow{6},\thisrow{7})}}] {\mytable};
    \foreach \i in {8,9}{
      \addplot table[x index=0,y index=\i] {\mytable};
    }
  \end{axis}
\end{tikzpicture}
\label{fig:er1}
}\hfill
\subfloat[Random sensor knn graph, size 500]{
\pgfplotstableread{knn_bandlim_sig1_500_8it50_snr.dat}\mytable

\begin{tikzpicture}
  \begin{axis}[
  	width = 0.33\linewidth,
    cycle list name=mydccol,
    ymin=6.2,
    ymax=10.3
    ]
    \pgfplotstablegetcolsof{\mytable}
    \pgfmathparse{\pgfplotsretval-1}
    \foreach \i in {1,2}{
      \addplot table[x index=0,y index=\i] {\mytable};
    }
    \addplot table[
    x index=0,
    y ={create col/expr={max(\thisrow{3},\thisrow{4},\thisrow{5},\thisrow{6},\thisrow{7})}}] {\mytable};
    \foreach \i in {8,9}{
      \addplot table[x index=0,y index=\i] {\mytable};
    }
  \end{axis}
  \label{fig:knn2}
\end{tikzpicture}
}\hfill \\
\subfloat[Scale-free graph, size 500]{
\pgfplotstableread{BA_bandlim_sig1_500_8it50_snr.dat}\mytable
\begin{tikzpicture}

  \begin{axis}[
  	width = 0.33\linewidth,
    cycle list name=mydccol,
    ymin=8.3,
    ymax=11.0
    ]
    \pgfplotstablegetcolsof{\mytable}
    \pgfmathparse{\pgfplotsretval-1}
    \foreach \i in {1,2}{
      \addplot table[x index=0,y index=\i] {\mytable};
    }
    \addplot table[
    x index=0,
    y ={create col/expr={max(\thisrow{3},\thisrow{4},\thisrow{5},\thisrow{6},\thisrow{7})}}] {\mytable};
    \foreach \i in {8,9}{
      \addplot table[x index=0,y index=\i] {\mytable};
    }
  \end{axis}
\end{tikzpicture}
\label{fig:ba2}
}\hfill
\subfloat[Community graph, size 500, 5 communities]{
\pgfplotstableread{comm_bandlim_sig1_500_10it50_2020_12_4_18_17_snr.dat}\mytable
\begin{tikzpicture}
\centering
  \begin{axis}[
  	width = 0.33\linewidth,
    cycle list name=mydccol,
    ymin=-8.5,
    ymax=12.0
    ]
    \pgfplotstablegetcolsof{\mytable}
    \pgfmathparse{\pgfplotsretval-1}
    \foreach \i in {1,2}{
      \addplot table[x index=0,y index=\i] {\mytable};
    }
    \addplot table[
    x index=0,
    y ={create col/expr={max(\thisrow{3},\thisrow{4},\thisrow{5},\thisrow{6},\thisrow{7})}}] {\mytable};
    \foreach \i in {8,9}{
      \addplot table[x index=0,y index=\i] {\mytable};
    }
  \end{axis}
\end{tikzpicture}
\label{fig:comm2}
}\hfill
\subfloat[WS graph, size 500]{
\pgfplotstableread{WS_bandlim_sig1_500_5_p2it50_snr.dat}\mytable
\begin{tikzpicture}

  \begin{axis}[
  	width = 0.33\linewidth,
    cycle list name=mydccol,
    ymin=8.3,
    ymax=10.6,
    ]
    \pgfplotstablegetcolsof{\mytable}
    \pgfmathparse{\pgfplotsretval-1}
    \foreach \i in {1,2}{
      \addplot table[x index=0,y index=\i] {\mytable};
    }
    \addplot table[
    x index=0,
    y ={create col/expr={max(\thisrow{3},\thisrow{4},\thisrow{5},\thisrow{6},\thisrow{7})}}] {\mytable};
    \foreach \i in {8,9}{
      \addplot table[x index=0,y index=\i] {\mytable};
    }
  \end{axis}
\end{tikzpicture}
\label{fig:ws2}
}\hfill \\
\subfloat[Erd\H{o}s R\'{e}nyi graph, size 500]{
\pgfplotstableread{ER_bandlim_sig1_500_p02it50_snr.dat}\mytable
\centering
\begin{tikzpicture}

  \begin{axis}[
  	width = 0.33\linewidth,
    cycle list name=mydccol,
    ymin=8.0,
    ymax=10.8
    ]
    \pgfplotstablegetcolsof{\mytable}
    \pgfmathparse{\pgfplotsretval-1}
    \foreach \i in {1,2}{
      \addplot table[x index=0,y index=\i] {\mytable};
    }
    \addplot table[
    x index=0,
    y ={create col/expr={max(\thisrow{3},\thisrow{4},\thisrow{5},\thisrow{6},\thisrow{7})}}] {\mytable};
    \foreach \i in {8,9}{
      \addplot table[x index=0,y index=\i] {\mytable};
    }
  \end{axis}
\end{tikzpicture}
}\hfill
\subfloat[Classification accuracies]{
\pgfplotstableread{set_avg.dat}\mytable
\begin{tikzpicture}
\centering
  \begin{axis}[
  	width = 0.33\linewidth,
    cycle list name=mydccol,
    ]
    \pgfplotstablegetcolsof{\mytable}
    \pgfmathparse{\pgfplotsretval-1}
    \foreach \i in {1,...,\pgfmathresult}{
      \addplot table[x index=0,y index=\i] {\mytable};
    }
  \end{axis}
\end{tikzpicture}
    \label{fig:classAcc}
}
\subfloat{
\begin{tikzpicture}
  \begin{axis}[
    width = 0.33\linewidth,
    xmin=50,
    xmax=200,
    ymin=-20,
    ymax=0,
    axis line style={draw=none},
    tick style={draw=none},
    yticklabels=\empty,
    xticklabels=\empty
    legend style={draw=white!15!black,at={(1,1)}},
    ]
    \addlegendimage{blue1, mark=square*}
    \addlegendentry{\AlgRand{}};
    \addlegendimage{violet1, mark=otimes*}
    \addlegendentry{\AlgSP{}};
    \addlegendimage{orange1, mark=triangle*}
    \addlegendentry{\AlgEDfree{}};
    \addlegendimage{lblue1, mark=diamond*}
    \addlegendentry{\AlgDC{}};
    \addlegendimage{red1, mark=mystar}
    \addlegendentry{\AlgAVM{}};
  \end{axis}
\end{tikzpicture}
\label{fig:legend1}
}\hfill
\subfloat{
\begin{tikzpicture}
  \begin{axis}[
    width = 0.33\linewidth,
    axis lines=none,
    xmin=50,
    xmax=200,
    ymin=-20,
    ymax=0,
    legend style={draw=white!15!black}
    ]
  \end{axis}
\end{tikzpicture}
}\hfill
\caption{Comparison of eigendecomposition-free methods in the  literature. x-axis: number of samples selected. y-axis: average SNR of the reconstructed signals. We do not include the entire range of SNR from \AlgRand{} based reconstruction because of its comparatively wider range.}
\label{fig:mainRecon}
\end{figure*}

We will now describe how algorithms are initialized.

\subsection{Initialization details}
We wish to evaluate all the algorithms on an equal footing. So for evaluating graph squared coherences using Function \ref{alg:cumCoh}, we use the same number of random vectors $10\log(n)$ corresponding to $c = 10$, $\epsilon = 0.01$, and an order 30 polynomial wherever filtering is needed for the \AlgRand{}, \AlgDC{}, and \AlgAVM{} methods. Larger $c$ and smaller $\epsilon$ values result in a more accurate approximation of squared coherences, but also require more computations. We choose those particular values to achieve a balance between the approximation accuracy and the amount of computations. The degree of the polynomial is selected so as to be larger than the diameter of most graphs we consider.

The various algorithms we consider have some hard-coded parameter values. \AlgSP{}  has just one parameter $k$ to which we assign $k=4$. \AlgEDfree{}  has a few more parameters to tune like $\nu,\eta$. In \cite{sakiyama2019eigendecomposition} the parameter $\nu$ is chosen experimentally to be 75, but in our experiments we run the \AlgEDfree{} algorithm on a wide range of $\nu$ values around 75 --- $\nu=[0.075, 7.5, 75, 750, 75000]$, and select  the value of $\nu$ that maximizes SNR. We chose this wide range of values experimentally, as we observed
cases where optimal reconstruction SNR was achieved at $\nu$ values differing from the proposed 75 by several orders of magnitude when we considered different topologies, graph sizes and Laplacians. As for the sampling times, we choose the sampling time corresponding to the $\nu$ chosen as per the maximum SNR. We experimentally determine $\eta$ the same way as in the original implementation. For the \AlgDC{} algorithm, we choose $\Delta=0.9$.

\subsection{Reconstruction techniques}
\label{sec:reconErrors}
We denote the sampled signal $\mathbf{f}_{\mathcal{S}}$ and the lowpass frequencies of the original signal  by $\tilde{\mathbf{f}}_{\recF{}} = \mathbf{U}_{\recF{}}^\transp \mathbf{f}$.
The ideal reconstruction which minimizes the mean square error using the sampled signal is given by the least squares solution to $\norm{\mathbf{U}_{\mathcal{S}\recF{}}\tilde{\mathbf{f}}_\recF{}-\mathbf{f}_\mathcal{S}}_2$. Other existing methods of reconstruction are --- using a linear combination of tailored kernels as seen in \AlgEDfree{}, or solving a regularized least squares as seen in \AlgGersh{}. However, since we are primarily interested in comparing the sampling sets generated by various algorithms on an even footing, we use the least squares solution for reconstruction which we compute by assuming that we know the graph Fourier basis.
To achieve best results for \AlgRand{}, instead of least squares solution we use the recommended weighted least squares \cite{puy2016random}, although it is slightly different from what we use for all other algorithms. We use Moore-Penrose pseudo inverse for all our least squares solutions.

\section{Results}
\label{sec:perfAlg}
We now evaluate the performance of our algorithm based on its speed and
on how well
it can reconstruct the original signal.

\begin{figure*}[t]
\subfloat[Random sensor graphs with 20 nearest neighbour connections]{
    \pgfplotstableread[col sep=comma]{2020_11_20_10_37_alg_times_iter_max_sel.dat}\mytable
    \centering
   \begin{tikzpicture}

  \begin{loglogaxis}[
  	width = 0.47\linewidth,
    cycle list name=my col,
    xlabel={Number of nodes in the graph},
    ylabel={Execution time (secs.)},
    legend pos=north west,
    legend entries = {\AlgRand{},\AlgSP{},\AlgEDfree{},\AlgGersh{},\AlgAVM{}}
    ]
    \legend{\AlgRand{},\AlgSP{},\AlgEDfree{},\AlgGersh{},\AlgAVM{}}
    \pgfplotstablegetcolsof{\mytable}
    \pgfmathparse{\pgfplotsretval-1}
    \foreach \i in {1,2,11,9,8}{
      \addplot table[x index=0,y index=\i] {\mytable};
    }
  \end{loglogaxis}
\end{tikzpicture}
\label{fig:timeAlg}
}\hfill
\subfloat[Community graphs with 10 communities]{
    \pgfplotstableread[col sep=comma]{2020_12_10_2_38_alg_times_iter_max_sel.dat}\mytable
    \centering
   \begin{tikzpicture}

  \begin{loglogaxis}[
  	width = 0.47\linewidth,
    cycle list name=my col,
    xlabel={Number of nodes in the graph},
    ylabel={Execution time (secs.)},
    legend pos=north west,
    legend entries = {\AlgRand{},\AlgSP{},\AlgEDfree{},\AlgGersh{},\AlgAVM{}}
    ]
    \legend{\AlgRand{},\AlgSP{},\AlgEDfree{},\AlgGersh{},\AlgAVM{}}
    \pgfplotstablegetcolsof{\mytable}
    \pgfmathparse{\pgfplotsretval-1}
    \foreach \i in {1,2,11,9,8}{
      \addplot table[x index=0,y index=\i] {\mytable};
    }
  \end{loglogaxis}
\end{tikzpicture}
\label{fig:timeAlgComm}
}
\caption{Visualizing average sampling times of four algorithms over 50 iterations on community graphs with 10 communities. Execution times for \AlgEDfree{} are averaged over executions for different parameter values.}
\end{figure*}
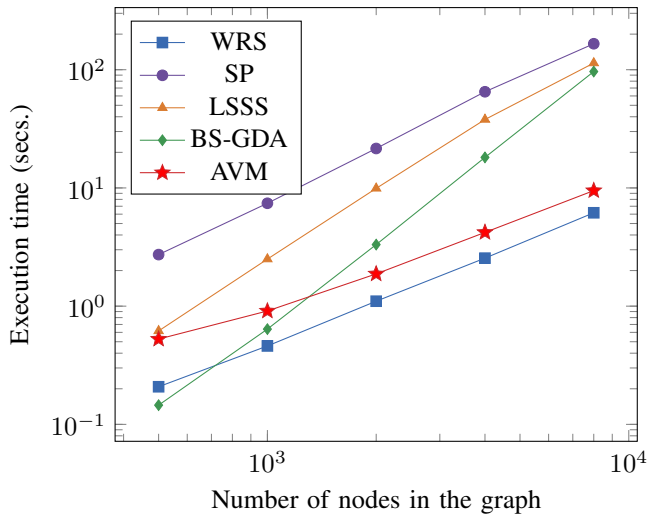
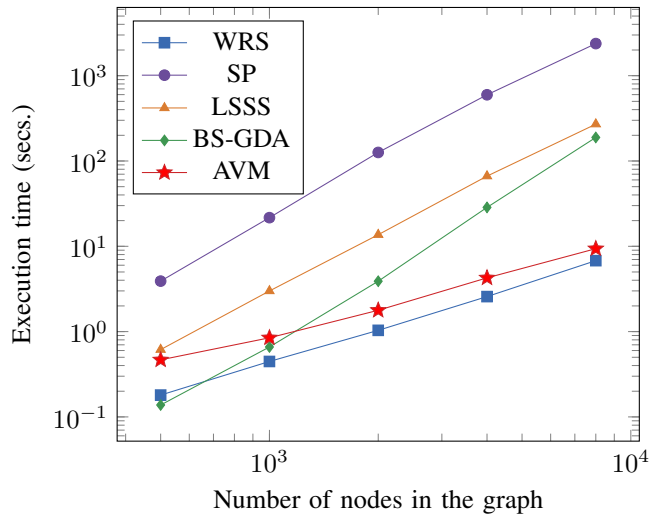

\subsection{Reconstruction error}
In this paper we look at the mean squared error in reconstructing the signal $\mathbf{f}$. So we measure the error $\norm{\hat{\mathbf{f}}-\mathbf{f}}^2$ in the reconstructed signal with respect to the original noisy signal, where $\hat{\mathbf{f}}$ is the reconstructed signal.

For our experiments, we plot the SNR averaged over 50 different graph instances of each model of graph, with a new signal generated for each graph instance Fig. \ref{fig:mainRecon}. The two graph models where we observe a lesser SNR for \AlgAVM{} algorithm are the Random sensor nearest neighbors and the Community graph models which we discuss next. For the remaining graph models the reconstruction SNR from \AlgDC{} and \AlgAVM{} sampling is comparable to other algorithms, such as \AlgSP{} and \AlgEDfree{}. In fact, we find the sampling from \AlgAVM{} to be comparatively satisfactory considering that we are reporting the maximum SNRs for \AlgEDfree{} over 5 different parameter values.

In Fig. \ref{fig:knn2} we notice that for Random sensor nearest neighbors graphs of size 500 we need more  samples to achieve competitive performance.
To better understand this, consider a graph consisting of two communities. In the original volume maximization a sampling set from only one community would give a volume of zero, and that sampling set would never be selected by a greedy exact volume maximization. However,  because \AlgAVM{}  is only an approximation to the greedy volume maximization,  sampling from only one community is possible although unlikely. More generally, this approximation affects weakly connected graphs such as random sensor graphs with a knn construction with a small number of nearest neighbors. More specifically, this approximation affects community graphs at low sampling rates as we can see in Figures \ref{fig:comm1} and \ref{fig:comm2}.

The issue, however, is no longer critical for larger graphs as we see when we increase the number of samples to 150 --- (Table  \ref{tab:reconScaleComm}), our algorithm performance is comparable to that of other state-of-the-art algorithms. Note that we do not face this issue for Erd\H{o}s R\'{e}nyi graph instances in our experiments, since these are almost surely connected as the probability of connection exceeds the sharp threshold \cite{erdHos1960evolution}.

For the USPS dataset classification, we do observe a significant drop in the classification accuracy for the samples chosen using the \AlgDC{} algorithm. However, the classification accuracy for \AlgAVM{} on the USPS dataset is at par with the remaining algorithms.

Next we evaluate how the complexity of \AlgAVM{}  scales with the graph size.

\subsection{Speed}
\label{sec:avm_speed}

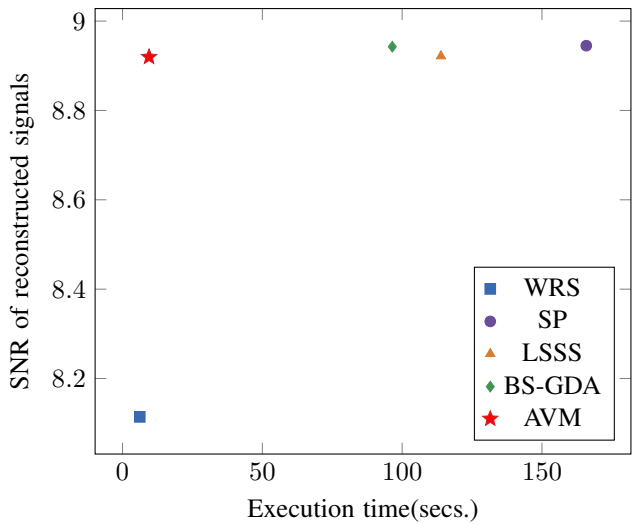
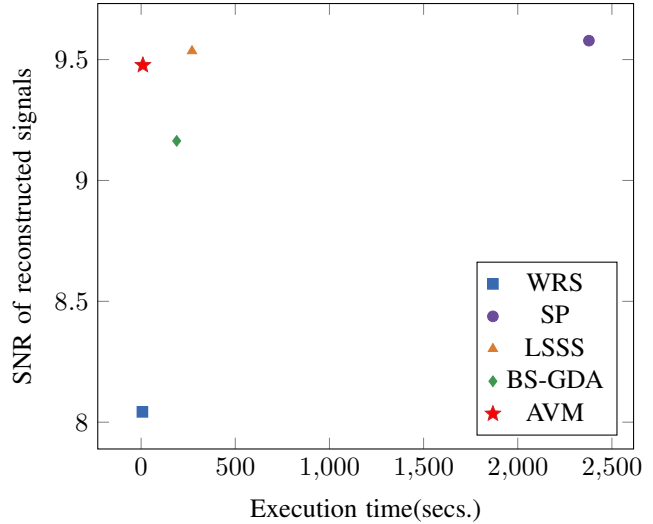
\begin{figure*}[t]
\subfloat[Random sensor graphs with 20 nearest neighbour connections.]{
   \begin{tikzpicture}

  \begin{axis}[
  	width = 0.48\linewidth,
    cycle list name=my col,
    xlabel={Execution time(secs.)},
    ylabel={SNR of reconstructed signals},
    legend pos=south east,
    ]
    \legend{\AlgRand{},\AlgSP{},\AlgEDfree{},\AlgGersh{},\AlgAVM{}}
    \pgfplotstableread{2020_11_20_10_37_sz_alg_snrs.dat}\mytablesnrs
    \pgfplotstableread[col sep=comma]{2020_11_20_10_37_alg_times_iter_max_sel.dat}\mytabletimes
    \pgfplotstablecreatecol[copy column from table={\mytablesnrs}{[index]1}]{\AlgRand{}}{\mytabletimes}
    \pgfplotstablecreatecol[copy column from table={\mytablesnrs}{[index]2}]{\AlgSP{}}{\mytabletimes}
    \pgfplotstablecreatecol[copy column from table={\mytablesnrs}{[index]3}]{GSSS1}{\mytabletimes}
    \pgfplotstablecreatecol[copy column from table={\mytablesnrs}{[index]4}]{GSSS2}{\mytabletimes}
    \pgfplotstablecreatecol[copy column from table={\mytablesnrs}{[index]5}]{GSSS3}{\mytabletimes}
    \pgfplotstablecreatecol[copy column from table={\mytablesnrs}{[index]6}]{GSSS4}{\mytabletimes}
    \pgfplotstablecreatecol[copy column from table={\mytablesnrs}{[index]7}]{GSSS5}{\mytabletimes}
    \pgfplotstablecreatecol[copy column from table={\mytablesnrs}{[index]8}]{AVM}{\mytabletimes}
    \pgfplotstablecreatecol[copy column from table={\mytablesnrs}{[index]9}]{RandomGilles}{\mytabletimes}
    \pgfplotstablecreatecol[copy column from table={\mytablesnrs}{[index]10}]{\AlgGersh{}}{\mytabletimes}
    \addplot+[
    only marks,
    select coords between index={4}
    ]
    table[x index=1,y index=20] {\mytabletimes};
    \addplot+[
    only marks,
    select coords between index={4}
    ]
    table[x index=2,y index=13] {\mytabletimes};
    \addplot+[
    only marks,
    select coords between index={4}
    ]
    table[x index=11,y index=14] {\mytabletimes};
    \addplot+[
    only marks,
    select coords between index={4}
    ]
    table[x index=9,y index=21] {\mytabletimes};
    \addplot+[
    only marks,
    select coords between index={4}
    ]
    table[x index=8,y index=19] {\mytabletimes};
  \end{axis}
\end{tikzpicture}
    \label{fig:tsnr_scatter}
}\hfill
\subfloat[Community graphs with 10 communities]{
   \begin{tikzpicture}

  \begin{axis}[
  	width = 0.48\linewidth,
    cycle list name=my col,
    xlabel={Execution time(secs.)},
    ylabel={SNR of reconstructed signals},
    legend pos=south east,
    ]
    \legend{\AlgRand{},\AlgSP{},\AlgEDfree{},\AlgGersh{},\AlgAVM{}}
    \pgfplotstableread{2020_12_10_2_38_sz_alg_snrs.dat}\mytablesnrs
    \pgfplotstableread[col sep=comma]{2020_12_10_2_38_alg_times_iter_max_sel.dat}\mytabletimes
    \pgfplotstablecreatecol[copy column from table={\mytablesnrs}{[index]1}]{\AlgRand{}}{\mytabletimes}
    \pgfplotstablecreatecol[copy column from table={\mytablesnrs}{[index]2}]{\AlgSP{}}{\mytabletimes}
    \pgfplotstablecreatecol[copy column from table={\mytablesnrs}{[index]3}]{GSSS1}{\mytabletimes}
    \pgfplotstablecreatecol[copy column from table={\mytablesnrs}{[index]4}]{GSSS2}{\mytabletimes}
    \pgfplotstablecreatecol[copy column from table={\mytablesnrs}{[index]5}]{GSSS3}{\mytabletimes}
    \pgfplotstablecreatecol[copy column from table={\mytablesnrs}{[index]6}]{GSSS4}{\mytabletimes}
    \pgfplotstablecreatecol[copy column from table={\mytablesnrs}{[index]7}]{GSSS5}{\mytabletimes}
    \pgfplotstablecreatecol[copy column from table={\mytablesnrs}{[index]8}]{AVM}{\mytabletimes}
    \pgfplotstablecreatecol[copy column from table={\mytablesnrs}{[index]9}]{RandomGilles}{\mytabletimes}
    \pgfplotstablecreatecol[copy column from table={\mytablesnrs}{[index]10}]{\AlgGersh{}}{\mytabletimes}
    \addplot+[
    only marks,
    select coords between index={4}
    ]
    table[x index=1,y index=20] {\mytabletimes};
    \addplot+[
    only marks,
    select coords between index={4}
    ]
    table[x index=2,y index=13] {\mytabletimes};
    \addplot+[
    only marks,
    select coords between index={4}
    ]
    table[x index=11,y index=14] {\mytabletimes};
    \addplot+[
    only marks,
    select coords between index={4}
    ]
    table[x index=9,y index=21] {\mytabletimes};
    \addplot+[
    only marks,
    select coords between index={4}
    ]
    table[x index=8,y index=19] {\mytabletimes};
  \end{axis}
\end{tikzpicture}
}
\caption{Scatter plot of the SNR vs execution time for graph size 8000.}
    \label{fig:tsnr_scatter_comm}
\end{figure*}

Using the setup from Section \ref{sec:expSetup}, we compare the sampling times for \AlgRand{}, \AlgSP{}, \AlgEDfree{}, \AlgGersh{} and \AlgAVM{} algorithms. We exclude the \AlgDC{} algorithm from these comparisons because the distance evaluations in \AlgDC{}, which provide good intuition, make \AlgDC{} significantly slower as compared to \AlgAVM{}. We include \AlgGersh{}  since it is one of the lowest complexity approaches among eigendecomposition-free algorithms. We use the Random sensor graph from the GSPbox \cite{perraudin2016gspbox} with 20 nearest neighbors.

\begin{table}[t]
\centering
    \caption{Execution time(secs.) for sampling, Random senor graphs}
    \pgfplotstableread[col sep=comma]{recon_iter_compare_2021_8_19_23_5_execution_times_combine.dat}\mytableTime
    \pgfplotstabletypeset[columns={[index]0,[index]1,[index]2,[index]3,[index]4,[index]5,[index]6},
    columns/0/.style={column name=$\abs{\mathcal{V}}$, fixed},
	columns/1/.style={column name=\AlgRand{}, fixed, precision=1},
	columns/2/.style={column name=\AlgSP{}, fixed, precision=1},
	columns/3/.style={column name=\AlgEDfree{}, fixed, precision=1},
	columns/4/.style={column name=\makecell[t]{BS-\\GDA}, fixed, precision=1},
	columns/5/.style={column name=\AlgAVM{}, fixed, precision=1},
	columns/6/.style={column name=\makecell[t]{Overhead\\(\AlgAVM)}},
	every head row/.style={before row=\hline,after row=\hline},
	every last row/.style={after row=\hline},
	empty cells with={--}]{\mytableTime}
    \label{tab:timeScale}
\end{table}
\begin{table}[t]
\centering
    \caption{SNRs, Random sensor graphs}
    \pgfplotstableread[col sep=comma]{recon_iter_compare_2021_8_19_23_5_snrs_combine.dat}\mytableTime
    \pgfplotstabletypeset[columns={[index]0,[index]1,[index]2,[index]3,[index]4,[index]5},
    columns/0/.style={column name=$\abs{\mathcal{V}}$, fixed},
	columns/1/.style={column name=\AlgRand{}},
	columns/2/.style={column name=\AlgSP{}},
	columns/3/.style={column name=\AlgEDfree{}},
	columns/4/.style={column name=\AlgGersh{}},
	columns/5/.style={column name=\AlgAVM{}},
	every head row/.style={before row=\hline,after row=\hline},
	every last row/.style={after row=\hline},
	empty cells with={--}]{\mytableTime}
    \label{tab:reconScale}
\end{table}
\begin{table}[t]
\centering
    \caption{Execution time(secs.) for sampling, Community graphs}
    \pgfplotstableread[col sep=comma]{recon_iter_compare_2021_9_4_4_31_execution_times_combine.dat}\mytableTime
    \pgfplotstabletypeset[columns={[index]0,[index]1,[index]2,[index]3,[index]4,[index]5,[index]6},
    columns/0/.style={column name=$\abs{\mathcal{V}}$, fixed},
	columns/1/.style={column name=\AlgRand{}, fixed, precision=1},
	columns/2/.style={column name=\AlgSP{}, fixed, precision=1},
	columns/3/.style={column name=\AlgEDfree{}, fixed, precision=1},
	columns/4/.style={column name=\AlgGersh{}, fixed, precision=1},
	columns/5/.style={column name=\AlgAVM{}, fixed, precision=1},
	columns/6/.style={column name=\makecell[t]{Overhead\\(\AlgAVM)}},
	every head row/.style={before row=\hline,after row=\hline},
	every last row/.style={after row=\hline},
	empty cells with={--}]{\mytableTime}
    \label{tab:timeScaleComm}
\end{table}
\begin{table}[t]
\centering
    \caption{SNRs, Community graphs}
    \pgfplotstableread[col sep=comma]{recon_iter_compare_2021_9_4_4_31_snrs_combine.dat}\mytableTime
    \pgfplotstabletypeset[columns={[index]0,[index]1,[index]2,[index]3,[index]4,[index]5},
    columns/0/.style={column name=$\abs{\mathcal{V}}$, fixed},
	columns/1/.style={column name=\AlgRand{}},
	columns/2/.style={column name=\AlgSP{}},
	columns/3/.style={column name=\AlgEDfree{}},
	columns/4/.style={column name=\AlgGersh{}},
	columns/5/.style={column name=\AlgAVM{}},
	every head row/.style={before row=\hline,after row=\hline},
	every last row/.style={after row=\hline},
	empty cells with={--}]{\mytableTime}
    \label{tab:reconScaleComm}
\end{table}

In our comparison we use the implementations of  \AlgRand{}, \AlgSP{}, \AlgEDfree{} and \AlgGersh{} distributed by their respective authors, and run them on  \softMat{} 2019b along with our proposed algorithm.  We could possibly improve on the existing implementations using specialized packages for functionalities such as  eigendecomposition, but to remain faithful to the original papers we use their codes with minimal changes. Wherever the theoretical algorithms in the papers conflict with the provided implementations we go with the implementation since that was presumably what the algorithms in the papers were timed on.

To minimize the effect of other processes running at different times, we run the sampling algorithms in a round robin fashion. We do this process for multiple iterations and different graph topologies. We time the implementations on an Ubuntu HP Z840 Workstation, which naturally has plenty of background processes running. The changes in their resource consumption affects our timing. It is impossible to stop virtually all background processes, so we try to reduce their impact in two ways. We iterate over each sampling scheme 50 times and report the averages. Secondly, instead of completing iterations over the sampling schemes one by one, we call all the different sampling schemes in the same iteration. These minor precautions help us mitigate any effects of background processes on our timing.

For 500-8,000 graph sizes, we observe that as the size of the graph increases, \AlgAVM{} is only slightly slower compared to \AlgRand{}. It is orders of magnitude faster than \AlgSP{}, \AlgEDfree{} and the \AlgGersh{} algorithm --- see Tables \ref{tab:timeScale} and \ref{tab:timeScaleComm}, while having a very small impact on the SNR of the reconstructed signal --- Tables \ref{tab:reconScale} and  \ref{tab:reconScaleComm}. The execution time also scales very well with respect to the graph size Fig. \ref{fig:timeAlg}.

We also report the relative execution times using the overhead rate, the ratio of execution times of two algorithms. We compute this overhead for the \AlgAVM{} algorithm vs the \AlgRand{} algorithm pair.
\begin{align*}
&\text{Overhead(\AlgAVM)} \\&= \frac{\text{Execution time of proposed \AlgAVM{} algorithm}}{\text{Execution time of \AlgRand{}}}
\end{align*}
Most existing graph sampling algorithms consider \AlgRand{} as the fastest sampling algorithm and benchmark against it. By specifying our overhead rates with respect to \AlgRand{} we can indirectly compare our algorithm with myriad others without doing so one by one. We report these factors for various graph sizes in Tables \ref{tab:timeScale} and \ref{tab:timeScaleComm}.

To justify the increase in the speed of execution compared to the slight decrease in the SNR, we plot the SNR versus the Execution time for the different algorithms we compared. Ideally we want an algorithm with fast execution and good SNR. From Fig. \ref{fig:tsnr_scatter}, \ref{fig:tsnr_scatter_comm} we see that our algorithm fits that requirement very well.

For experiments on graph sizes 500-8000, we reduced the variability in the execution time and SNR observations by reporting means over 50 randomly initialized graph and signal realizations. However, for graphs the size of 50,000 and 100,000, running 50 realizations of each sampling strategy is impractical because of the time required. To determine if 10 realizations are sufficient, we compute the ratio of standard deviation to the mean for execution times and SNRs. Except for one setting of \AlgGersh{} where the ratio is 0.26, it does not exceed 0.12 in all experiments. So for graphs of size 50,000 and 100,000, we report the execution times and SNRs averaged over 10 randomly initialized realizations in Tables \ref{tab:timeScale}, \ref{tab:reconScale}, \ref{tab:timeScaleComm}, and \ref{tab:reconScaleComm}.

In those tables of algorithm comparisons, we look for scalable algorithms that have low execution times and high SNRs, or which at least finish execution within our limits as mentioned in Section \ref{sec:effectGraphSizes}. For graphs with size 50,000, \AlgRand{}, \AlgEDfree{}, \AlgGersh{}, and \AlgAVM{}, finish within our limits, while for graphs with size 100,000 only  \AlgRand{}, and \AlgAVM{} can finish. We fill the table entries corresponding to the algorithms that did not finish with a $-$. Among the algorithms that finish, \AlgAVM{} provides up to $20\%$ improvement in the SNR from reconstructed signal over that of \AlgRand{}, and at most $0.46\%$ decrease compared to other algorithms, although the SNRs are low compared to smaller graph sizes because of the \ProjCS{}-based reconstruction. The execution times of \AlgAVM{} are at least $60\%$ less and as much as $93\%$ less compared to other state-of-the-art algorithms except \AlgRand{}. The overhead of \AlgAVM{} relative to \AlgRand{} is larger compared to smaller graph sizes because of the 5000 sampled vertices for 50,000 and 100,000 graph sizes as opposed to 150 samples for 500 to 8,000 graph sizes. We see a further increase of about $62\%$ in the execution time for community graphs due to the larger number of edges. So for graph sizes 50,000 and 100,000, \AlgAVM{} not only finishes execution within our limits, but maintains SNR at par with other algorithms for two different graph topologies, while being the fastest algorithm second only to \AlgRand{}.

Of course with different graph types, the SNR vs Execution time performance of \AlgAVM{} may vary. But what we always expect this algorithm to deliver is execution times similar to \AlgRand{} while having a significant improvement in the SNR. In a way, this algorithm bridges the gap between existing Eigendecomposition-free algorithms and \AlgRand{}.

\begin{figure}[t]
      \pgfplotstableread[col sep=comma]{recon_iter_compare_time_vs_samples_useful_2021_7_23_21_4_exchanged_alg_times_vs_samples.dat}\timeVsSamplesSmall
    \centering
   \begin{tikzpicture}
  \begin{loglogaxis}[
  	width = 0.97\linewidth,
    cycle list name=my col,
    xlabel={Number of samples},
    ylabel={Execution time (secs.)},
    legend pos=north west,
    legend entries = {\AlgRand{}, \AlgSP{}, \AlgEDfree{}, \AlgGersh{}, \AlgAVM{}}
    ]
    \pgfplotstablegetcolsof{\timeVsSamplesSmall}
    \pgfmathparse{\pgfplotsretval-1}
    \foreach \i in {1, 2, 3, 4, 5}{
      \addplot table[x index=0,y index=\i] {\timeVsSamplesSmall};
    }
  \end{loglogaxis}
\end{tikzpicture}
    \caption{Random sensor graphs with 8192 vertices and 20 nearest neighbour connections.}
    \label{fig:randSenExeNSampAll}
\end{figure}
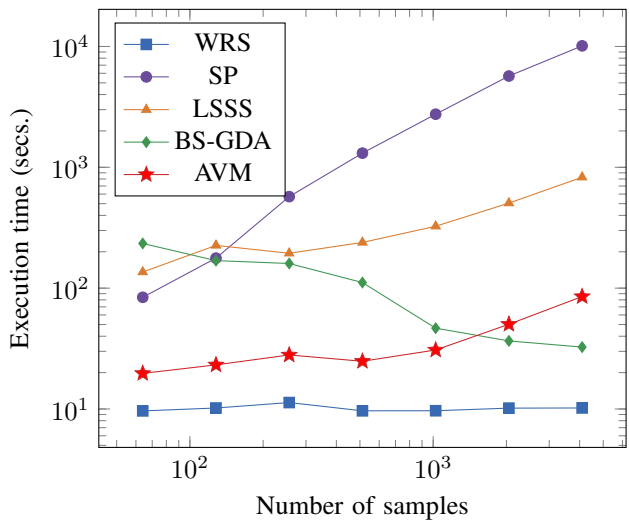
\begin{figure}[t]
\pgfplotstableread[col sep=comma]{recon_iter_compare_time_vs_samples_useful_2021_7_24_20_19_exchanged_alg_times_vs_samples.dat}\timeVsSamplesSmall
    \centering
   \begin{tikzpicture}

  \begin{loglogaxis}[
  	width = 0.97\linewidth,
    cycle list name=my col,
    xlabel={Number of samples},
    ylabel={Execution time (secs.)},
    legend pos=north west,
    legend entries = {\AlgRand{}, \AlgSP{}, \AlgEDfree{}, \AlgGersh{}, \AlgAVM{}}
    ]
    \pgfplotstablegetcolsof{\timeVsSamplesSmall}
    \pgfmathparse{\pgfplotsretval-1}
    \foreach \i in {1, 2, 3, 4, 5}{
      \addplot table[x index=0,y index=\i] {\timeVsSamplesSmall};
    }
  \end{loglogaxis}
\end{tikzpicture}
\caption{Community graphs with 8192 vertices and 10 communities.}
    \label{fig:commExeNSampAll}
\end{figure}
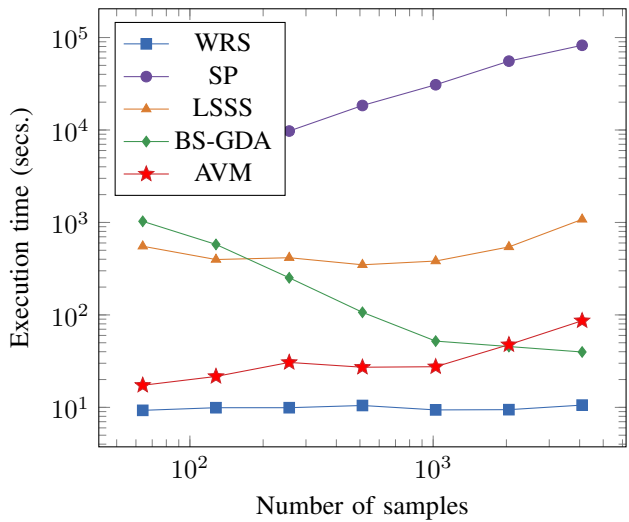

\begin{figure}[t]
    \pgfplotstableread[col sep=comma]{recon_iter_compare_time_vs_samples_useful_2021_7_23_21_15_exchanged_alg_times_vs_samples.dat}\timeVsSamplesLarge
    \centering
   \begin{tikzpicture}

  \begin{axis}[
  	width = 0.97\linewidth,
    cycle list name=myavmcol,
    xlabel={Number of samples},
    ylabel={Execution time (secs.)},
    legend pos=south east,
    legend entries = {\AlgAVM{}}
    ]
    \pgfplotstablegetcolsof{\timeVsSamplesLarge}
    \pgfmathparse{\pgfplotsretval-1}
    \foreach \i in {3}{
      \addplot table[x index=0,y index=\i] {\timeVsSamplesLarge};
    }
  \end{axis}
\end{tikzpicture}
    \caption{Random sensor graphs with 50,000 vertices and 20 nearest neighbour connections.}
    \label{fig:randSenExeNSampAVM}
\end{figure}
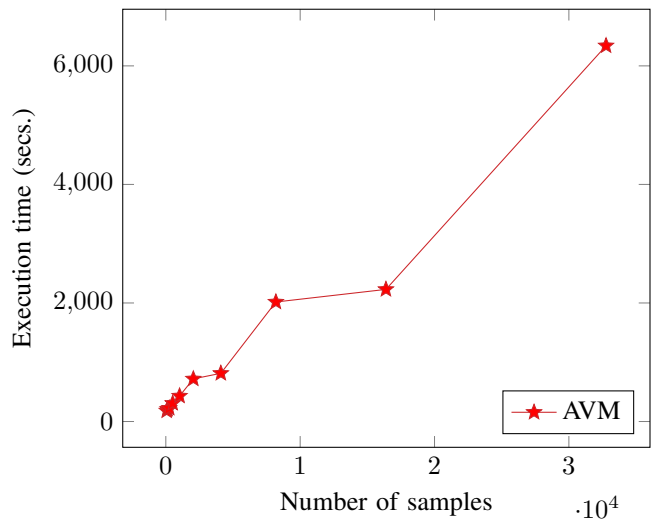
\begin{figure}[t]
      \pgfplotstableread[col sep=comma]{recon_iter_compare_time_vs_samples_useful_2021_7_24_19_24_exchanged_alg_times_vs_samples.dat}\timeVsSamplesLarge
    \centering
   \begin{tikzpicture}

  \begin{axis}[
  	width = 0.97\linewidth,
    cycle list name=myavmcol,
    xlabel={Number of samples},
    ylabel={Execution time (secs.)},
    legend pos=south east,
    legend entries = {\AlgAVM{}},
    scaled y ticks=false]
    \pgfplotstablegetcolsof{\timeVsSamplesLarge}
    \pgfmathparse{\pgfplotsretval-1}
    \foreach \i in {3}{
      \addplot table[x index=0,y index=\i] {\timeVsSamplesLarge};
    }
  \end{axis}
\end{tikzpicture}
    \caption{Community graphs with 50,000 vertices and 10 communities.}
    \label{fig:commExeNSampAVM}
\end{figure}
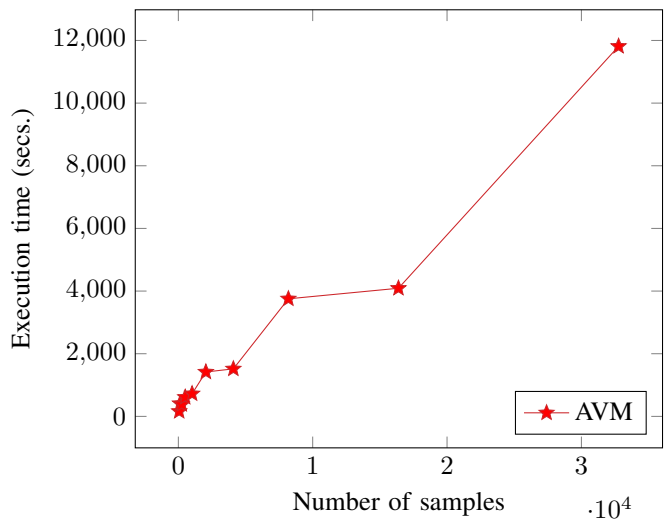

\subsection{Effect of number of samples on execution time}
A fast graph signal sampling algorithm should be scalable with respect to the number of sampled vertices. To experiment and compare the scalability of different graph signal sampling algorithms, we set up sensor graphs with 20 nearest neighbors, and community graphs with 10 communities, both of size 8192. For each of the graph types, we sample a varying number of vertices ranging from 64 to 4096 samples in multiples of 2 and measure the corresponding execution times. We display the results of this scalability experiment as execution times versus the number of samples in plots.

Doing this experiment for different sampling methods helps us compare their robustness to a varying number of samples. In figures \ref{fig:randSenExeNSampAll} and \ref{fig:commExeNSampAll}, we observe the effect of varying the number of sampled vertices on the execution times of the algorithms \AlgRand{}, \AlgSP{}, \AlgEDfree{}, \AlgGersh{}, and \AlgAVM{}. With an increase in the number of sampled vertices, \AlgRand{}'s execution time does not show a significant dependence, \AlgSP{} and \AlgEDfree{}'s execution time increases, whereas \AlgGersh{}'s execution time decreases. The minimal dependence of \AlgRand{}'s execution time is due to the computationally cheap random sampling step once the graph coherences are computed by Function \ref{alg:cumCoh}. The increase in the execution time of \AlgSP{} and \AlgEDfree{} is due to the are extra computations needed for sampling a new vertex. However for \AlgGersh{}, we believe that the decrease in the execution time is due to the decrease in the coverage set sizes with the increase in the number of requested samples. We see that \AlgAVM{}'s execution time is close to \AlgRand{}'s for smaller sampling sets, and it increases with an increase in the number of requested samples. Except for sample set sizes the order of graph size, \AlgAVM{} has the second-lowest execution times for a range of sample set sizes.

Since we saw that the execution time of \AlgAVM{} increases with the number of sampled vertices, we wish to further assess the rate of increase of the execution time. For this purpose, we consider sensor graphs with 20 nearest neighbors and community graphs with 10 communities, both of size 50,000. The number of samples requested range from 64 samples to 32768 samples in multiples of 2. We display the results of this experiment as execution times versus the number of samples plots limited to \AlgAVM{}.

In Figures \ref{fig:randSenExeNSampAVM} and \ref{fig:commExeNSampAVM}, apart from minor fluctuations, we see that the execution times vs the number of samples data points lie on a straight line for an order of $10^4$ range in the number of sampled vertices. This observation agrees with our theoretical analysis of \AlgAVM{} complexity in Section \ref{sec:complexNSamp} explaining an additional $O(s\abs{\mathcal{E}}d)$ dependence on the number of samples compared to \AlgRand{}.

\section{Conclusion}
Most sampling schemes perform reasonably well when dealing with perfectly bandlimited signals. However, in the presence of noise or the signal not being perfectly bandlimited, some schemes perform much better. In the scenario that only a limited number of samples can be chosen, we would like to use an algorithm that can perform well without requiring computationally expensive procedures such as eigendecomposition.

The algorithms presented in this paper rely on the intuition of looking at the problem as maximizing the volume of the parallelepiped formed by the lowpass signals corresponding to the sampled vertices. This helps us to develop intuitive and fast graph signal sampling algorithms. The volume maximization framework also helped to connect various existing algorithms.

The sampling algorithm we developed reaches speeds achieved by \AlgRand{}, but with a large improvement in reconstruction accuracies. The accuracies are comparable with other contemporary algorithms but at the same time provide significant improvements in speed.

\section{Acknowledgements}
This work is supported in part by NSF under grants CCF-1410009, CCF-1527874, and CCF-2009032 and by a gift from Tencent.

\appendices
\numberwithin{equation}{section}
\counterwithin{figure}{section}
\counterwithin{table}{section}
\section{Proof of eigenvector convergence}
\begin{lemma}
There exists a signal $\boldsymbol{\phi}$ in the orthogonal subspace to $\boldsymbol{\psi}^*$ with $\boldsymbol{\phi}(\Smiter{})=\mathbf{0}, \norm{\boldsymbol{\phi}}=1$ whose out of bandwidth energy is a minimum value $c_0\neq 0$.
\end{lemma}
\begin{proof}
The set of signals $\{\boldsymbol{\phi}:\boldsymbol{\phi}(\Smiter{})=\mathbf{0}, \norm{\boldsymbol{\phi}}=1\}$ is a closed set. Let $\begin{pmatrix} x_1 & \cdots & x_n\end{pmatrix}^\transp$ be in the set for any $\epsilon$. Then $\begin{pmatrix} x_1 + \epsilon/2 & x_2 & \cdots & x_n\end{pmatrix}^\transp$ is in the $\epsilon$ neighborhood. Distance exists because it is a normed vector space. That vector does not have $\norm{}=1$ so it is not in the set. So for every $\epsilon$-neighborhood $\exists$ a point not in the set. So every point is a limit point and the set is a closed set.

Out of bandwidth energy is a continuous function on our set. Let $\mathbf{v}_1, \mathbf{v}_2$ be such that $\mathbf{v}_1, \mathbf{v}_2 \perp \boldsymbol{\psi}^*$ and $\mathbf{v}_1(\Smiter{}) = \mathbf{0}, \mathbf{v}_2(\Smiter{}) = \mathbf{0}$. Let us suppose that the Fourier coefficients for $\mathbf{v}_1$ and $\mathbf{v}_2$ are $(\alpha_1, \cdots, \alpha_n)^\transp$ and $(\beta_1, \cdots, \beta_n)^\transp$. Then we want
\begin{equation}
    \sum_{i=m+2}^n (\alpha_i^2 - \beta_i^2) < \epsilon
    \label{eq:ene_diff}
\end{equation}
for some $\delta$ where $\norm{\mathbf{v}_1-\mathbf{v}_2}<\delta$. We can show that \eqref{eq:ene_diff} holds when $\delta=\epsilon/2$.

Since the set is closed and the function is continuous on the set, the function attains a minimum value. Minimum value cannot be zero because there is a unique signal $\boldsymbol{\psi}^*$ with that property, and we are looking in a space orthogonal to $\boldsymbol{\psi}^*$. So there is a signal with minimum out of bandwidth energy of $c_0$ where $c_0>0$.
\end{proof}
\label{app:conv}
Next, this appendix shows the proof for the $l_2$ convergence from Theorem \ref{thm:eigConverge}.
\begin{proof}
Let us look at a particular step where we have already selected $\Smiter{}$ vertices. The solutions to the following optimization problems are equivalent.
\begin{equation*}
    \boldsymbol{\psi}^*_k = \argminA_{\boldsymbol{\psi}} \frac{\boldsymbol{\psi}^\transp \mathbf{L}^k \boldsymbol{\psi}}{\boldsymbol{\psi}^\transp \boldsymbol{\psi}} = \argminA_{\boldsymbol{\psi}, \norm{\boldsymbol{\psi}}=1} \boldsymbol{\psi}^\transp \mathbf{L}^k \boldsymbol{\psi}.
\end{equation*}
Therefore, we will consider solutions with $\norm{\boldsymbol{\psi}}=1$.

Let us consider the space of our signals. $\boldsymbol{\phi}(\Smiter{})=\mathbf{0}, \boldsymbol{\phi}\in \mathcal{R}^n$ is a vector space. Dimension of this vector space is $n-m$.

For any $k$, let us represent our solution for $k$ as $\boldsymbol{\psi} = \alpha_1 \boldsymbol{\psi}^* + \alpha_2 \boldsymbol{\psi}^\perp$. Here $\boldsymbol{\psi}^\perp$ is a vector in the orthogonal subspace to our vector $\boldsymbol{\psi}^*$. We can do this because we have a vector space and it has finite dimensions. One condition on our signal is that $\alpha_1^2 + \alpha_2^2 = 1$, $\norm{\boldsymbol{\psi}^*}=1, \norm{\boldsymbol{\psi}^\perp}=1$. Furthermore, we know the Fourier transform of our two signal components.
\begin{align*}
    \boldsymbol{\psi}^* \xrightarrow{\mathcal{F}} &\mathbf{U}^\transp\boldsymbol{\psi}^* = \begin{bmatrix}
    \gamma_1&
    \cdots&
    \gamma_{m+1}&
    0&
    \cdots&
    0
    \end{bmatrix}^\transp = \boldsymbol{\gamma},\\
    \boldsymbol{\psi}^\perp \xrightarrow{\mathcal{F}} &\mathbf{U}^\transp\boldsymbol{\psi}^\perp = \begin{bmatrix}
    \beta_1&
    \cdots&
    \beta_n&
    \end{bmatrix}^\transp = \boldsymbol{\beta}.
\end{align*}
Our signal can be written as
\begin{equation*}
    \begin{bmatrix}
    \boldsymbol{\psi}^* & \boldsymbol{\psi}^\perp
    \end{bmatrix}\begin{bmatrix}
    \alpha_1\\
    \alpha_2
    \end{bmatrix} = \begin{bmatrix}
    \boldsymbol{\psi}^* & \boldsymbol{\psi}^\perp
    \end{bmatrix} \boldsymbol{\alpha}.
\end{equation*}
Our objective function becomes the following:
\begin{equation*}
\begin{aligned}
    &\boldsymbol{\alpha}^\transp \begin{bmatrix}
    \boldsymbol{\psi}^{*T} \\ \boldsymbol{\psi}^{\perp T}
    \end{bmatrix} \mathbf{L}^k \begin{bmatrix}
    \boldsymbol{\psi}^* & \boldsymbol{\psi}^\perp
    \end{bmatrix} \boldsymbol{\alpha}\\&= \boldsymbol{\alpha}^\transp \begin{bmatrix}
    \boldsymbol{\gamma}^\transp \\ \boldsymbol{\beta}^\transp
    \end{bmatrix}\boldsymbol{\Sigma}^k \begin{bmatrix}
    \boldsymbol{\gamma} & \boldsymbol{\beta}
    \end{bmatrix} \boldsymbol{\alpha}\\
    \\&= \boldsymbol{\alpha}^\transp \begin{bmatrix}
    \sum_{i=1}^{m+1} \gamma_i^2 \sigma_i^k & \sum_{i=1}^{m+1} \gamma_i \beta_i \sigma_i^k\\
    \sum_{i=1}^{m+1} \gamma_i \beta_i \sigma_i^k & \sum_{i=i}^n\beta_i^2 \sigma_i^k
    \end{bmatrix} \boldsymbol{\alpha}= \boldsymbol{\alpha}^\transp \begin{bmatrix}
    a & b\\
    b & d
    \end{bmatrix} \boldsymbol{\alpha}.
\end{aligned}
\end{equation*}
In the last equation $a, b, c, d$ are just convenient notations for the scalar values in the $2\times 2$ matrix. Note that $a,d>0$ because $\sigma_i>0$ and the expression is then a sum of positive quantities. Note that we want to minimize the objective function subject to the constraint $\norm{\alpha}=1$. We solve this optimization problem by a standard application of the KKT conditions \cite{boyd2004convex}.

The Lagrangian function corresponding to our constrained minimization problem is as follows:
\begin{equation*}
    L(\boldsymbol{\alpha}, \lambda) = \boldsymbol{\alpha}^\transp \begin{bmatrix}
    a & b\\
    b & d
    \end{bmatrix} \boldsymbol{\alpha} + \lambda(\boldsymbol{\alpha}^\transp \boldsymbol{\alpha} - 1).
\end{equation*}
The solution which minimizes this objective function is the eigenvector of the matrix $\begin{bmatrix} a & b \\ c & d\end{bmatrix}$ with the minimum eigenvalue. To prove that we take the gradient of the equation with respect to $\boldsymbol{\alpha}$ and put it to $\mathbf{0}$.

This gives us two first order equations.
\begin{equation*}
    a\alpha_1 + b\alpha_2 + \lambda \alpha_1 = 0, \qquad
    b\alpha_1 + d\alpha_2 + \lambda \alpha_2 = 0.
\end{equation*}
$\alpha_1 = 0$ implies $\alpha_2 = 0$ unless $b= 0$ and vice versa. Both $\alpha_1$ and $\alpha_2$ cannot be zero at the same time otherwise our solution does not lie in our domain of unit length vectors. However either of $\alpha_1$ or $\alpha_2$ can be $0$ only if $b=0$. If $b=0$, for large $k$ the solution is given by $\alpha_2=0, \alpha_1=1$ because we show next that $a/d$ can be made less than $1/2$ for $k>k_0$. We now analyze the case where $b\neq 0$ and so $a+\lambda \neq 0$ and $d+\lambda \neq 0$. Writing $\alpha_2$ in terms of $\alpha_1$ for both the equations we get:
\begin{equation*}
    \alpha_2 = \frac{-(a+\lambda)}{b}\alpha_1,\quad \alpha_2 = \frac{-b}{d+\lambda}\alpha_1.
\end{equation*}
Equating both the expressions for $\alpha_2$ (and assuming $\alpha_1\neq 0$)gives us a quadratic with two solutions. Since $a,d>0$ the positive sign gives us the $\lambda$ with lower magnitude.
\begin{equation*}
    \lambda = \frac{-(a+d)+ \sqrt{(a-d)^2 + 4 b^2}}{2}.
\end{equation*}
We now use the condition that the solution has norm one. Solution of this gives us a value for $\alpha_2$.
\begin{equation*}
\alpha_2 = \mp \frac{a-d + \sqrt{(a-d)^2 + 4b^2}}{\sqrt{\roundB{a-d + \sqrt{(a-d)^2+4b^2}}^2 + 4b^2}}.
\end{equation*}
We look at the absolute value of the $\alpha_2$.
\begin{equation*}
    \abs{\alpha_2} = \frac{2\abs{b}}{\sqrt{(-(a-d)+\sqrt{(a-d)^2+4b^2})^2+(2b)^2}}
\end{equation*}

Let us find a $k_0$ such that $\abs{b}/d<\epsilon/2$($\epsilon>0$) for all $k>k_0$. This will also make $a/d<1/2$. This will help us make the entire expression less than $\epsilon$ for $k>k_0$ \eqref{eq:kvalA}.
We upper bound $b$ in the following way.
\begin{align*}
    \abs{b} &= \abs{\sum_{i=1}^{m+1} \beta_i \gamma_i \sigma^k}\\
    &\leq \sqrt{\sum_{i=1}^{m+1}\beta_i^2 \gamma_i^2} \sqrt{\sum_{i=1}^{m+1}\sigma_i^{2k}}\\
    &\leq 1. \sqrt{m+1} \sigma_{m+1}^k = b_1.
\end{align*}
Now we know that the least possible value of $d$ is $c_0 \sigma_{m+2}^k$. So when $k>k_0$ we get the following upper bound for $\abs{b}/d$ in terms of $\epsilon$:
\begin{equation*}
    \frac{\abs{b}}{d} \leq \frac{\sqrt{m+1} \sigma_{m+1}^k}{c_0 \sigma_{m+2}^k}.
\end{equation*}
 We want this to be less than $\epsilon/2$, which gives us our condition on $k$.
 \begin{align*}
     \frac{\sqrt{m+1} \sigma_{m+1}^k}{c_0 \sigma_{m+2}^k} < \frac{\epsilon}{2}\\
 k > \ceil*{\frac{\log(m+1)/2+ \log 1/\epsilon + \log(2/c_0)}{\log\frac{\sigma_{m+2}}{\sigma_{m+1}}}}. \numberthis \label{eq:kval}
 \end{align*}
$a/d$ also admits a similar analysis.
\begin{align*}
    \frac{a}{d} &\leq \frac{\sigma_{m+1}^k}{c_0 \sigma_{m+2}^k}\\
    k &> \ceil*{\frac{\log (2/c_0)}{\log\frac{\sigma_{m+2}}{\sigma_{m+1}}}}. \numberthis \label{eq:kvalA}
\end{align*}

Since this value of $k$ is equal or lesser than the value of $k$ required for $\abs{b}/d<\epsilon/2$, for our theorem we will take the value \eqref{eq:kval}.
When $d$ divides both the numerator and denominator of the equation it gives us the expressions we need in terms of $a/d$ and $\abs{b}/d$.
\begin{align*}
    \abs{\alpha_2}&=\frac{\abs{2b/d}}{\sqrt{\roundB{1-a/d+\sqrt{(a/d-1)^2 + 4(b/d)^2}}^2 + (2b/d)^2}}\\
    &< \frac{\epsilon}{\abs{1-a/d + \sqrt{(1-a/d)^2+\epsilon^2}}}\\
    &< \frac{\epsilon}{\abs{2(1-a/d)}}< \epsilon.
\end{align*}

This implies that as $k$ increases the coefficient of out-of-bandwidth component goes to zero. Because the out-of bandwidth signal has finite energy, the signal energy goes to zero as $\alpha_2 \rightarrow 0$. Whether $\boldsymbol{\psi}^*_k$ converges to $\boldsymbol{\psi}^*$ or $-\boldsymbol{\psi}^*$ is a matter of convention. Hence as $k \rightarrow \infty$, $\boldsymbol{\psi}^*_k \rightarrow \boldsymbol{\psi}^*$.
\end{proof}

\section{Justification for ignoring target bandwidth while sampling}
\label{sec:ignoreBand}
We know that selecting the right $\mathbf{U}_{\mathcal{S}\recF{}}$ matrix is essential to prevent a blow-up of the error while reconstructing using \eqref{eq:recon}. In practice for reconstruction the bandwidth is $\bwidth{}\leq s$. However, for our \AlgAVM{} sampling algorithm we chose the bandwidth to be $\abs{\recR{}} = s$ instead. We next address why that is a logical choice with respect to D-optimality.

The matrix $\mathbf{U}_{\mathcal{S}\mathcal{S}}^\transp \mathbf{U}_{\mathcal{S}\mathcal{S}}$ is positive definite following from our initial condition of the set $\mathcal{S}$ being a uniqueness set. This provides us with the needed relations between determinants. We can see that  $\mathbf{U}_{\mathcal{S}\recF{}}^\transp \mathbf{U}_{\mathcal{S}\recF{}}$ is a submatrix of $\mathbf{U}_{\mathcal{S}\mathcal{S}}^\transp \mathbf{U}_{\mathcal{S}\mathcal{S}}$.
\begin{equation*}
    \mathbf{U}_{\mathcal{S}\mathcal{S}}^\transp \mathbf{U}_{\mathcal{S}\mathcal{S}} = \begin{bmatrix}\mathbf{U}_{\mathcal{S}\recF{}}^\transp \mathbf{U}_{\mathcal{S}\recF{}} & \mathbf{U}_{\mathcal{S}\recF{}}^\transp \mathbf{U}_{\mathcal{S},\bwidth{}+1:s}\\ \mathbf{U}_{\mathcal{S},\bwidth{}+1:s}^\transp\mathbf{U}_{\mathcal{S}\recF{}} & \mathbf{U}_{\mathcal{S},\bwidth{}+1:s}^\transp\mathbf{U}_{\mathcal{S},\bwidth{}+1:s}\end{bmatrix}.
\end{equation*}
This helps us to relate the matrix and its submatrix determinants using Fischer's inequality from Theorem 7.8.5 in \cite{horn2012matrix}.
\begin{equation}
    \det(\mathbf{U}_{\mathcal{S}\mathcal{S}}^\transp \mathbf{U}_{\mathcal{S}\mathcal{S}}) \leq \det(\mathbf{U}_{\mathcal{S}\recF{}}^\transp \mathbf{U}_{\mathcal{S}\recF{}})\det(\mathbf{U}_{\mathcal{S},\bwidth{}+1:s}^\transp \mathbf{U}_{\mathcal{S},\bwidth{}+1:s})
    \label{eq:hadamardDetIneq}
\end{equation}

The determinant of the matrix $\mathbf{U}_{\mathcal{S},\bwidth{}+1:s}^\transp\mathbf{U}_{\mathcal{S},\bwidth{}+1:s}$ can be bounded above. The eigenvalues of $\mathbf{U}_{\mathcal{S},\bwidth{}+1:s}^\transp \mathbf{U}_{\mathcal{S},\bwidth{}+1:s}$ are the same as the non-zero eigenvalues of $\mathbf{U}_{\mathcal{S},\bwidth{}+1:s}\mathbf{U}_{\mathcal{S},\bwidth{}+1:s}^\transp$ by Theorem 1.2.22 in \cite{horn2012matrix}. Using eigenvalue interlacing Theorem 8.1.7 from \cite{golub2012matrix}, the eigenvalues of the matrix $\mathbf{U}_{\mathcal{S},\bwidth{}+1:s}\mathbf{U}_{\mathcal{S},\bwidth{}+1:s}^\transp$ are less than or equal to 1 because it is submatrix of $\mathbf{U}_{\mathcal{V},\bwidth{}+1:s}\mathbf{U}_{\mathcal{V},\bwidth{}+1:s}^\transp$ whose non-zero eigenvalues are all $1$. As the determinant of a matrix is the product of its eigenvalues, the following bound applies:
\begin{equation}
    \det(\mathbf{U}_{\mathcal{S},\bwidth{}+1:s}^\transp\mathbf{U}_{\mathcal{S},\bwidth{}+1:s}) \leq 1.
    \label{eq:outDetBound}
\end{equation}

Using \eqref{eq:hadamardDetIneq} and \eqref{eq:outDetBound} and positive definiteness of the matrices, we now have a simple lower bound for our criteria under consideration:
\begin{equation}
    \abs{\det(\mathbf{U}_{\mathcal{S}\mathcal{S}}^\transp \mathbf{U}_{\mathcal{S}\mathcal{S}})} \leq \abs{\det(\mathbf{U}_{\mathcal{S}\recF{}}^\transp \mathbf{U}_{\mathcal{S}\recF{}})}.
    \label{eq:detBoundLarge}
\end{equation}
Thus, for example it is impossible for $\abs{\det(\mathbf{U}_{\mathcal{S}\mathcal{S}}^\transp \mathbf{U}_{\mathcal{S}\mathcal{S}})}$ to be equal to some positive value while $\abs{\det(\mathbf{U}_{\mathcal{S}\recF{}}^\transp \mathbf{U}_{\mathcal{S}\recF{}})}$ being half of that positive value.

To summarize, instead of aiming to maximize $\abs{\det(\mathbf{U}_{\mathcal{S}\recF{}}^\transp\allowbreak \mathbf{U}_{\mathcal{S}\recF{}})}$, we aimed to maximize $\abs{\det(\mathbf{U}_{\mathcal{S}\mathcal{S}}^\transp \mathbf{U}_{\mathcal{S}\mathcal{S}})}$. This intuitively worked because optimizing for a D-optimal matrix indirectly ensured a controlled performance of the subset of that matrix.
In this way due to the relation \eqref{eq:detBoundLarge}, we avoided knowing the precise bandwidth $\bwidth{}$ and still managed to sample using the \AlgAVM{} algorithm.

\section{Approximating Gram matrix by a diagonal matrix}
\label{sec:approxGramDiag}
Here we try to estimate how close our approximation of $\lpDmam^\transp \lpDmam{}$ to a diagonal matrix is. Towards this goal we define a simple metric for a general matrix $\mathbf{A}$.
\begin{equation}
    \text{Fraction of energy in diagonal} = \frac{\sum_i \mathbf{A}_{ii}^2}{\sum_i\sum_j \mathbf{A}_{ij}^2}.
    \label{eq:diagMetric}
\end{equation}
Since this can be a property dependent on the graph topology, we take 5 different types of graphs with 1000 vertices --- Scale-free, \AlgRand{} sensor nearest neighbors, Erd\H{o}s R\'{e}nyi, Grid, Line. Using \AlgAVM{} we select a varying number of samples ranging from 1 to 50. With the bandwidth $\bwidth{}$ taken to be 50, we average the fraction of the energy \eqref{eq:diagMetric} over 10 instances of each graph and represent it in Fig. \ref{fig:testDiag}.

\begin{figure}
    \centering
    \begin{tikzpicture}
    \begin{axis}[xlabel={Number of samples},ylabel={Fraction of energy in diagonal},legend entries={BA,Random sensor knn,ER,Grid,Line},cycle list name=my col]
  \addplot table [x index=0,y index=1]{test_orthogonal.dat};
  \addplot table [x index=0,y index=2]{test_orthogonal.dat};
  \addplot table [x index=0,y index=3]{test_orthogonal.dat};
  \addplot table [x index=0,y index=4]{test_orthogonal.dat};
  \addplot table [x index=0,y index=5]{test_orthogonal.dat};
  \end{axis}
    \end{tikzpicture}
    \caption{Closeness to diagonal at each iteration.}
    \label{fig:testDiag}
\end{figure}
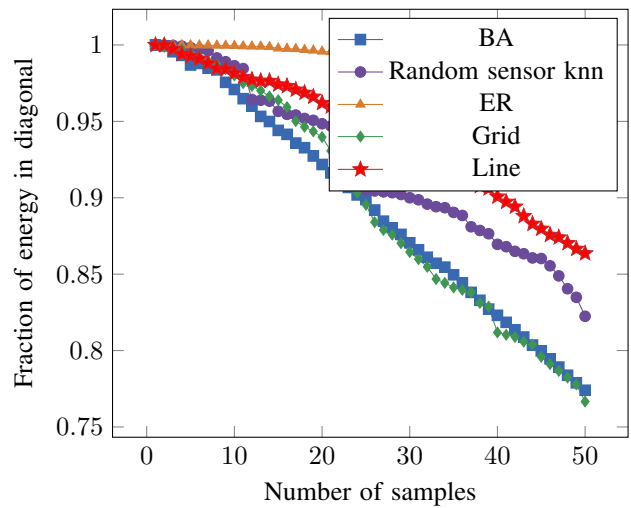

We observe more than 0.75 fraction of energy in the diagonal of the matrix $\lpDmam^\transp\lpDmam{}$, which justifies this approximation. According to our experiments, which are not presented here, the inverse of the matrix $(\lpDmam^\transp \lpDmam{})^{-1}$ is not as close to a diagonal matrix as $\lpDmam^\transp \lpDmam{}$ is to a diagonal matrix. Nevertheless, in place of $(\lpDmam^\transp \lpDmam{})^{-1}$ we still use $\diag\roundB{1/{\norm{\mathbf{d}_1}^2},\cdots,1/{\norm{\mathbf{d}_m}^2}}$ for what it is, an approximation.

Note however that the approximation does not hold in general for any samples. It holds when the samples are selected in a determinant maximizing conscious way by Algorithm  \ref{alg:avm}. This approximation is suited to \AlgAVM{} because of its choice of sampling bandwidth, $\recR$. As the number of samples requested increases, so does the sampling bandwidth. The higher bandwidth causes the filtered delta signals to become more localized causing energy concentration in the diagonal and keeping the diagonal approximation reasonable and applicable.

\section{D-optimal sampling for generic kernels}
\label{app:genKern}
Another graph signal model is a probabilistic distribution instead of a bandlimited model \cite{gadde2015probabilistic}, \cite{zhu2003semi}. In such cases, the covariance matrix is our kernel. The subset selection problem is defined as a submatrix selection of the covariance matrix. Framing the problem as entropy maximization naturally leads to a determinant maximization approach \cite{shewry1987maximum}.

To define our problem more formally, let us restrict space of all possible kernels to the space of kernels which can be defined as $\mathbf{K} = g(\mathbf{L})$ with $g$ defined on matrices but induced from a function from non-negative reals to positive reals $g: \mathbb{R}_{\geq 0} \rightarrow \mathbb{R}_{> 0}$. An example of such a function of $\mathbf{L}$ would be $(L+\delta I)^{-1}$. Such a function can be written as a function on the eigenvalues of the Laplacian $\mathbf{K}= \mathbf{U} g(\boldsymbol{\Sigma}) \mathbf{U}^\transp$. Motivated by entropy maximization in the case of probabilistic graph signal model, suppose we wish to select a set $\mathcal{S}$ so that  we maximize the determinant magnitude $\abs{\det(\mathbf{K}_{\mathcal{S}\mathcal{S}})}$.

There are a few differences for solving the new problem, although most of Algorithm \ref{alg:avm} translates well. We now wish to maximize $\abs{\det(\mathbf{U}_{\mathcal{S}}g(\boldsymbol{\Sigma})\mathbf{U}_{\mathcal{S}}^\transp)}$. The expression for the determinant update remains the same as before.
\begin{align*}
    &\det\roundB{\begin{bmatrix} \lpDmam^\transp \lpDmam{} & \lpDmam^\transp \mathbf{d}_v\\ \mathbf{d}_v^\transp \lpDmam{} & \mathbf{d}_v^\transp \mathbf{d}_v\end{bmatrix}} \\ &\approx\det(\lpDmam^\transp \lpDmam{})\det(\mathbf{d}_v^\transp \mathbf{d}_v - \mathbf{d}_v^\transp \lpDmam{} (\lpDmam^\transp\lpDmam{})^{-1}\lpDmam^\transp \mathbf{d}_v)
\end{align*}
Only now we have to maximize the volume of the parallelelpiped formed by the vectors $\mathbf{d}_v = \mathbf{U}g^{1/2}(\boldsymbol{\Sigma})\mathbf{U}^\transp\boldsymbol{\delta}_v$ for $v\in \mathcal{S}$. The squared coherences with respect to our new kernel $\mathbf{d}_v^\transp \mathbf{d}_v$ are computed in the same way as before by random projections. The diagonal of our new kernel matrix now approximates the matrix  $\lpDmam^\transp \lpDmam{}$.
\begin{equation*}
    \lpDmam^\transp \lpDmam{} \approx \diag((\mathbf{U}g(\boldsymbol{\Sigma})\mathbf{U}^\transp)_{11},\cdots,(\mathbf{U}g(\boldsymbol{\Sigma})\mathbf{U}^\transp)_{nn})
\end{equation*}
The other difference is that the approximate update stage is given by
\begin{equation*}
    v* \gets \argmaxA_{v\in \mathcal{S}^c} \norm{\mathbf{d}_v}^2 - \sum_{w\in \mathcal{S}}\frac{(\mathbf{U}g^{1/2}(\boldsymbol{\Sigma})\mathbf{U}^\transp\mathbf{d}_w)^2(v)}{\norm{\mathbf{d}_w}^2}
\end{equation*}
with the difference resulting from the kernel not being a projection operator. So for a generic kernel with a determinant maximization objective, Algorithm \ref{alg:avm} works the same way with minor modifications discussed here.

\bibliographystyle{hieeetr}
\apptocmd{\sloppy}{\hbadness 10000\relax}{}{}
\bibliography{graphsrefs}

\end{document}